\documentclass[journal]{IEEEtran}

\usepackage{cite}
\usepackage{amsmath,amssymb,amsfonts}
\usepackage{algorithmic}
\usepackage{setspace}
\usepackage{graphicx}
\usepackage{textcomp}
\usepackage{xcolor}
\usepackage{graphicx}
\usepackage{subfigure}
\usepackage[justification=centering]{caption}
\usepackage{verbatim}
\usepackage{balance}
\usepackage{multirow}
\usepackage{tabularx}
\usepackage{booktabs}  
\usepackage{caption}   
\usepackage{balance}  
\newlength{\mytabwidth}
\newtheorem{theorem}{Theorem}
\newtheorem{proof}{Proof}
\newtheorem{remark}{Remark}
\usepackage{hyperref}

\usepackage[ruled, linesnumbered, vlined]{algorithm2e}
\usepackage{booktabs}

\usepackage{array}

\newcommand{\PreserveBackslash}[1]{\let\temp=\\#1\let\\=\temp}

\newcolumntype{C}[1]{>{\PreserveBackslash\centering}p{#1}}
\newcolumntype{R}[1]{>{\PreserveBackslash\raggedleft}p{#1}}
\newcolumntype{L}[1]{>{\PreserveBackslash\raggedright}p{#1}}

\makeatletter
\renewcommand{\maketag@@@}[1]{\hbox{\m@th\normalsize\normalfont#1}}%
\makeatother

\def\BibTeX{{\rm B\kern-.05em{\sc i\kern-.025em b}\kern-.08em
    T\kern-.1667em\lower.7ex\hbox{E}\kern-.125emX}}

\begin{document}

\title{Road Network-Aware Personalized Trajectory Protection with Differential Privacy under Spatiotemporal Correlations\\}

\author{Minghui Min,~\IEEEmembership{Senior Member,~IEEE,} Jiahui Liu, Mingge Cao, Shiyin Li,
    \\ Hongliang Zhang,~\IEEEmembership{Member,~IEEE,}  Miao Pan,~\IEEEmembership{Senior Member,~IEEE,}
	and Zhu Han,~\IEEEmembership{Fellow,~IEEE}
	
	\IEEEcompsocitemizethanks{\IEEEcompsocthanksitem Corresponding author is Minghui Min. This work was partly presented at IEEE WCNC 2024 \cite{cao2024protecting}.
\IEEEcompsocthanksitem Minghui Min is with the School of Information and Control Engineering, China University of Mining and Technology, Xuzhou 221116, China, and also with the Key Laboratory of Aerospace Information Security and Trusted Computing, Ministry of Education and School of CyberScience and Engineering, Wuhan University, Wuhan 430072, China. E-mail: minmh@cumt.edu.cn.
		\IEEEcompsocthanksitem Jiahui Liu, Mingge Cao, and Shiyin Li are with the School of Information and Control Engineering, China University of Mining and Technology, Xuzhou 221116, China. E-mails: liujiahui@cumt.edu.cn, caomg@cumt.edu.cn, and lishiyin@cumt.edu.cn.
		\IEEEcompsocthanksitem Hongliang Zhang is with the School of Electronics, Peking University, Beijing 100871, China. E-mails: hongliang.zhang92@gmail.com.
		\IEEEcompsocthanksitem Miao Pan is with the Department of Electrical and Computer Engineering at the University of Houston, Houston, TX 77004 USA. E-mail: mpan2@uh.edu.
\IEEEcompsocthanksitem Zhu Han is with the Department of Electrical and Computer Engineering at the University of Houston, Houston, TX 77004 USA, and also with the Department of Computer Science and Engineering, Kyung Hee University, Seoul, South Korea, 446-70. E-mail: hanzhu22@gmail.com.}
}

\maketitle

\begin{abstract}
Location-Based Services (LBSs) offer significant convenience to mobile users but pose significant privacy risks, as attackers can infer sensitive personal information through spatiotemporal correlations in user trajectories. Since users' sensitivity to location data varies based on factors such as stay duration, access frequency, and semantic sensitivity, implementing personalized privacy protection is imperative. This paper proposes a Personalized Trajectory Privacy Protection Mechanism (PTPPM) to address these challenges. Our approach begins by modeling an attacker’s knowledge of a user’s trajectory spatiotemporal correlations, which enables the attacker to identify possible location sets and disregard low-probability location sets. To combat this, we integrate geo-indistinguishability with distortion privacy, allowing users to customize their privacy preferences through a configurable privacy budget and expected inference error bound. This approach provides the theoretical framework for constructing a Protection Location Set (PLS) that obscures users' actual locations. Additionally, we introduce a Personalized Privacy Budget Allocation Algorithm (PPBA), which assesses the sensitivity of locations based on trajectory data and allocates privacy budgets accordingly. This algorithm considers factors such as location semantics and road network constraints. Furthermore, we propose a Permute-and-Flip mechanism that generates perturbed locations while minimizing perturbation distance, thus balancing privacy protection and Quality of Service (QoS). Simulation results demonstrate that our mechanism outperforms existing benchmarks, offering superior privacy protection while maintaining user QoS requirements.

\end{abstract}


\section{Introduction}

Spatiotemporal trajectory data plays a crucial role in enhancing Location-Based Services (LBSs) for applications in Vehicular Ad Hoc Networks (VANETs) and other mobile networks, like smart cities, traffic management, and personal location services \cite{10177933, wang2019protecting, lv2021private}.
However, the use of such data raises significant privacy concerns, as it can potentially expose sensitive information about individuals' daily routines and behaviors, leading to risks of surveillance and misuse of personal data \cite{tang2021dlp, lu2011pseudonym, li2018location}.
In VANETs, user mobility is primarily confined to road networks \cite{tan2020protecting}, but many existing privacy protection schemes fail to fully account for the structural attributes of these networks. This gap allows the attacker with knowledge of the road network to launch more precise attacks. Thus, it is crucial to consider spatiotemporal correlations among different locations, as well as the geographical context of the road network, to ensure robust trajectory privacy protection.

Additionally, different users may exhibit unique requirements for trajectory privacy protection \cite{min20213d, jin2021ulpt}. The sensitivity of each location along a trajectory is influenced by factors such as stay duration, access frequency, and semantic sensitivity. Consequently, users’ privacy needs vary. Furthermore, these demands evolve based on users’ changing roles. For example, a bus driver who routinely operates along a specific route might consider these locations less sensitive due to their job compared to passengers who infrequently traverse the same routes. Therefore, ensuring the privacy of user trajectories while addressing personalized demands is essential.

Most existing research focuses on privacy protection for individual locations \cite{li2017achieving, sun2023synthesizing}. For instance, geo-indistinguishability, an extension of differential privacy, aims to protect a user’s location within a certain radius, guaranteeing ``generalized differential privacy" \cite{andres2013geo}. However, this approach does not account for adversaries’ arbitrary prior knowledge, leading to potential privacy leakage \cite{shokri2012protecting, yao2022privacy}, and the degree of privacy protection is often not clearly defined. To address these limitations, authors in \cite{yu2017dynamic} combined geo-indistinguishability with expected inference error to propose a personalized location privacy protection mechanism. However, this approach still does not consider the spatiotemporal correlations between different locations along a trajectory.

A solution based on “$\delta$-location set” differential privacy is proposed in \cite{xiao2015protecting}, combining the location privacy protection mechanism of \cite{yu2017dynamic} with temporal correlations between locations on a trajectory. However, this approach assigns the same privacy budget to all locations and does not accommodate users' personalized demands. Our previous work \cite{cao2024protecting} provides personalized privacy protection for different locations in a trajectory, taking into account the temporal correlation of locations according to specific privacy requirements of users. However, it lacks a comprehensive consideration of the factors that determine the sensitivity of locations and the spatial correlations of locations have not been included. 
A privacy protection mechanism that dynamically adjusts privacy levels based on the sensitivity of trajectory points is proposed in \cite{10666272} to enhance flexibility. However, spatiotemporal correlations between locations are not explicitly considered, and personalization is primarily based on location semantic sensitivity.
Some trajectory privacy protection schemes apply $k$-anonymity \cite{xing2021location} to generalize and aggregate individual trajectory data, ensuring that the trajectory is protected by mixing it with at least $k-$1 other trajectories. These methods, however, rely on a trusted third party and fail to offer strict privacy guarantees \cite{jin2022survey, wu2023tcpp}.  In summary, existing trajectory privacy protection schemes often lack comprehensive approaches in critical areas, such as considering the spatiotemporal correlations between locations on a trajectory, meeting users’ personalized needs, and ensuring the protection of users' actual locations without relying on a trusted third party. As a result, there is a need for a novel trajectory privacy mechanism that addresses these issues simultaneously.

In this paper, we propose a Personalized Trajectory Privacy Protection Mechanism (PTPPM) that takes into account the spatiotemporal correlations between locations along a trajectory. The mechanism constructs a location transition probability matrix, combined with prior probability correlations at each moment along the trajectory and road network constraints, to derive the potential location set for the user at each timestamp. By integrating geo-indistinguishability \cite{andres2013geo} and distortion privacy \cite{shokri2011quantifying}, we enhance the system’s robustness to location inference attacks while providing a personalized approach. Geo-indistinguishability limits the attacker’s posterior knowledge but cannot quantify the similarity between the inferred and actual locations, while distortion privacy ensures that the attacker’s expected inference error exceeds a certain threshold. The combination of these techniques strengthens the defense against location inference attacks. To further enhance privacy, we incorporate a minimum distance search algorithm based on the Hilbert curve \cite{yu2017dynamic} to identify a Protection Location Set (PLS) that encompasses all potential locations along the trajectory.

Moreover, the mechanism adjusts privacy settings through two privacy parameters to achieve personalized protection. We introduce a Personalized Privacy Budget Allocation Algorithm (PPBA), which assigns varying privacy budgets to different locations based on their sensitivity. This sensitivity is determined by factors such as stay duration, access frequency, and semantic sensitivity. Additionally, we establish a directed graph that captures spatiotemporal correlations and road network constraints, allowing for the allocation of privacy budgets based on the proximity of sensitive locations and neighboring nodes.

Lastly, we developed a Permute-and-Flip (PF) mechanism \cite{mckenna2020permute}, originally designed for data privacy protection, as a location perturbation mechanism. This novel adaptation achieves smaller perturbation distances, improving the balance between privacy and Quality of Service (QoS). Simulation results demonstrate that PTPPM offers personalized trajectory privacy protection and delivers superior privacy protection compared to PIVE \cite{yu2017dynamic}, PIM \cite{xiao2015protecting}, and CTDP \cite{10666272} while maintaining QoS requirement. The main contributions of our work are as follows:
\begin{enumerate}
\item We propose a Personalized Trajectory Privacy Protection Mechanism (PTPPM) designed to defend against attackers who exploit spatiotemporal correlations between various locations on a trajectory while also considering the road network constraints.

\item We develop a Personalized Privacy Budget Allocation Algorithm (PPBA) by integrating a more practical and comprehensive sensitivity assessment approach. By integrating the concepts of geo-indistinguishability and distortion privacy, we enhance the system's robustness against location inference attacks while achieving a personalized approach.

\item We present a novel Permute-and-Flip location perturbation mechanism, which achieves smaller perturbation distances, improving the balance between trajectory privacy and QoS.

\item We conduct comprehensive simulations to analyze the impact of different privacy budgets and expected inference errors on personalized user requirements. We also demonstrate that PTPPM outperforms benchmarks in terms of balancing trajectory privacy and QoS loss.
\end{enumerate}

The remainder of this paper is organized as follows. Section \ref{Related Work} introduces the related work. Section \ref{System Model} presents the system model. We present the trajectory privacy protection statement in Section \ref{Trajectory Privacy Protection Statement}. A PTPPM framework is proposed in Section \ref{Personalized Trajectory Privacy Protection Mechanism}. The evaluation results are provided in Section \ref{Simulation Results}. Finally, we conclude this work in Section \ref{Conclusion}.

\section{Related Work}\label{Related Work}
Existing trajectory privacy protection methods include generalization\cite{gramaglia2021glove}, mixed regions\cite{hou2021tracking}, and $k$-anonymity\cite{xing2021location}. These techniques ensure that individual trajectory data is generalized and aggregated to protect the true trajectory when combined with at least $k-$1 other trajectories into an anonymous region. However, these methods rely on trusted third parties, posing a risk of privacy leakage if the server is compromised or attacked. They also fail to provide strict privacy guarantees. In the face of background knowledge attacks and combination attacks, an attacker can still obtain personal privacy information. To overcome background knowledge attacks, differential privacy (DP) is first proposed in \cite{dwork2006differential}. DP uses strict mathematical theory to protect user privacy, ensuring that the attacker cannot infer personal information through background knowledge analysis.

Current fake trajectory privacy protection schemes do not fully consider factors such as actual terrain and road network conditions when generating fake trajectories, neglecting the spatiotemporal correlation of locations along a trajectory \cite{qiu2023novel}. The virtual rotation algorithm is proposed in \cite{7918561}, where actual trajectories are hidden among fake trajectories. However, only the spatial attributes of the trajectory data are considered, without considering the spatiotemporal correlations of different locations in the trajectory. The temporal correlation of different trajectory locations is addressed in \cite{xiao2015protecting}, which employs differential privacy concepts to obscure actual locations within a set of locations that incorporate prior knowledge, thereby decreasing the attacker's inference accuracy. A method for measuring trajectory privacy based on temporal correlation is further introduced in \cite{cao2018quantifying}, although it overlooks the spatial correlation between different locations. Building on this, spatiotemporal event privacy is proposed in \cite{cao2019protecting}, integrating both spatial and temporal correlations between trajectory points, and an algorithm is developed to quantify this privacy. However, these approaches do not adequately account for the geographical constraints between sequential locations.

Different users have varying privacy protection needs, and to accommodate personalized privacy preferences, different privacy budgets should be allocated for different sensitive locations \cite{xu2020personalized}. Various techniques, such as anonymization, fake locations, and mixed regions, have been proposed. A multi-level personalized $k$-anonymity privacy protection model is introduced in \cite{qian2024multi}, which personalizes the anonymization process for each granularity space using dynamic $k$-value sequences to meet diverse privacy needs. However, this approach relies on a trusted third party, which cannot provide stringent privacy guarantees. 
A location perturbation mechanism based on regional sensitivity is proposed in \cite{10666272}, where the privacy budget is dynamically allocated according to the sensitivity of trajectory points. However, the method does not consider the spatiotemporal correlation between trajectory.
Additionally, methods based on anonymization and suppression are vulnerable to the attacker with background knowledge.

\setlength{\textfloatsep}{10pt plus 1.0pt minus 2.0pt}
\begin{figure}[!t]
\begin{center}
\includegraphics[width=0.4\textwidth]{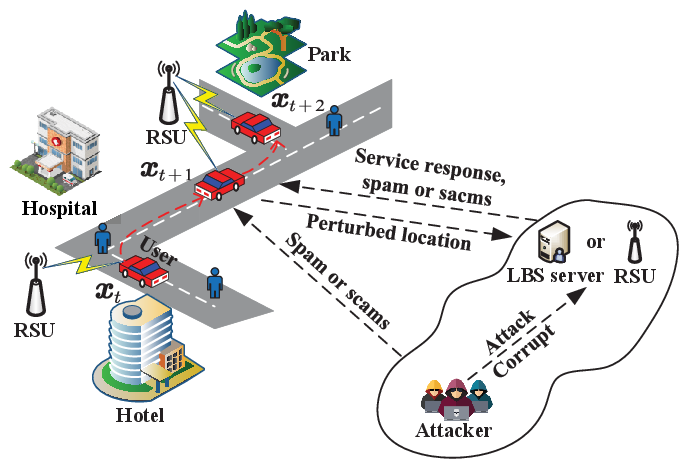}
\end{center}
\captionsetup{skip=-2pt}
\caption{Illustration of the trajectory privacy protection.}
\label{fig1}
\end{figure}

Differential privacy and expected inference error have a complementary relationship, and their combination can create more robust privacy protection mechanisms suitable for practical applications \cite{wang2020sparse, 10325612}.
Based on this, a personalized location privacy protection mechanism is proposed in \cite{yu2017dynamic}, which combines geo-indistinguishability \cite{andres2013geo} with expected inference error \cite{wang2020sparse} to meet the personalized privacy requirements. The Hilbert curve is used to search PLS because it has the advantage of reducing dimension and maintaining locality.  It can achieve better performance compared with other space-filling methods \cite{10325612}.
To counter powerful attacker who obtains the spatiotemporal correlations between different locations in a trajectory, the “$\delta$-location set” is proposed \cite{xiao2015protecting}. This method uses Markov chains to represent the temporal relationships between locations, protecting the actual locations at different timestamps. However, it applies uniform perturbation to all locations, disregarding personalized user needs. Building on this, our previous work in \cite{cao2024protecting} considers the temporal correlation of locations, providing personalized privacy protection for different locations in a trajectory based on user-specific requirements. Nonetheless, this mechanism only analyzes the impact of privacy parameters on the degree of personalized privacy protection within a simplified model without fully considering the factors determining location sensitivity. 

\setlength{\textfloatsep}{10pt plus 1.0pt minus 2.0pt}
\begin{figure}[!t]
\begin{center}
\includegraphics[width=0.34\textwidth]{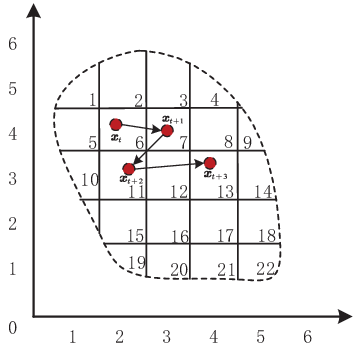}
\end{center}
\captionsetup{skip=-2pt}
\caption{User map coordinates and status coordinates.}
\label{fig2}
\end{figure}

\section{System Model}\label{System Model}
To obtain real-time LBSs, we consider that VANET users or mobile users share their location information with a roadside unit (RSU) or an LBS server at different times and locations \cite{min20213d}, \cite{9201413}. The user interacts with the RSU to obtain road information that can be used to search for nearby gas stations and navigate. The user's driving trajectory, illustrated in Fig. \ref{fig1}, showcases locations $\boldsymbol{x}_t$, $\boldsymbol{x}_{t+1}$, and $\boldsymbol{x}_{t+2}$ at different time points. Though capable of providing service responses, this server may compromise user privacy by probing into visited locations. Besides, the external attacker might corrupt or attack the untrusted LBS server to steal the user's location information. Through data analysis, the untrusted LBS server or external attackers could deduce the user's actual locations by leveraging spatiotemporal patterns and road network constraints. Ultimately, the attacker can use the inferred location data to send targeted spam or scams to the user for financial gain or other illegal purposes.
Important symbols are summarized in Table \ref{tab:1}.

\subsection{User Model}
We consider a VANET user driving within a city. We divide map $\mathcal{A}$ of the city into regions $\boldsymbol{x}_i$, each representing a semantic location, and each region has a 2D coordinate corresponding to it. $\boldsymbol{x}_i$ is a unit vector with the $i$-th element being 1 and other elements being 0. $\boldsymbol{x}_t$ represents the user's true location at time $t$, and $\boldsymbol{l}_t$ represents the two-dimensional coordinates of the user's location state at time $t$.
$Tr=\left[ \boldsymbol{x}_t,\boldsymbol{x}_{t+1},...,\boldsymbol{x}_{t+T} \right]$ represents the user's trajectory in continuous time length $T$. For example, as shown in Fig. \ref{fig2}, $\mathcal{A}=\left\{ \boldsymbol{x}_1,\ \boldsymbol{x}_2,\ \cdot \cdot \cdot ,\ \boldsymbol{x}_{22} \right\}$, $\boldsymbol{x}_t=\boldsymbol{x}_6=\left[ 0,0,0,0,0,1,0 \cdot \cdot \cdot ,0 \right]$, $\boldsymbol{l}_t=\left[ 2,4 \right]$, $Tr=\left[ \boldsymbol{x}_6,\boldsymbol{x}_7,\boldsymbol{x}_{11},\boldsymbol{x}_{13} \right]$.

The user uses the location perturbation mechanism to remap the actual location $\boldsymbol{x}_t$ from the actual location set $O_1$ to the fake location $\boldsymbol{x}_{t}^{'}$ from the perturbed location set $O_2$. The location perturbation probability distribution $f$ is given by
\begin{equation}\label{2}
f\left( \boldsymbol{x}_{t}^{'}|\boldsymbol{x}_t \right) =\text{Pr}\left( O_2=\boldsymbol{x}_{t}^{'}|O_1=\boldsymbol{x}_t \right) ,\ \ \ \ \boldsymbol{x}_t,\boldsymbol{x}_{t}^{'}\in \mathcal{A}.
\end{equation}
We use $\mathbf{p}_t$ to represent the user's location state at time $t$, where $\mathbf{p}_t\left[ i \right] =\text{Pr}\left( \boldsymbol{x}_t=\boldsymbol{x}_i \right) =\text{Pr}\left( \boldsymbol{l}_t \right)$ represents the probability that the user's actual location is in $\boldsymbol{x}_i$ at time $t$. Assuming that the user is distributed with the same probability $\mathcal{A}=\left\{ \boldsymbol{x}_2,\ \boldsymbol{x}_3,\ \boldsymbol{x}_5,\ \boldsymbol{x}_7 \right\}$, the location probability distribution of the user is $\mathbf{p}_t=\left[ 0,0.25,0.25,0,0.25,0,0.25,0,\cdot \cdot \cdot ,0 \right]$. We use $\mathbf{p}_{t}^{-}$ and $\mathbf{p}_{t}^{+}$ to represent the prior and posterior probabilities of the user before and after observing the released perturbed location $\boldsymbol{x}_{t}^{'}$.

\subsection{Attack Model}
We consider the attacker to be an untrusted LBS server or an external attacker who may attack or corrupt the LBS server. The attacker can access the user's current location information for commercial profit or illegal purposes. We assume that the attacker knows the location perturbation probability distribution $f\left( \boldsymbol{x}_{t}^{'}|\boldsymbol{x}_t \right)$, and can obtain the prior distribution $\mathbf{p}_{t}^{-}=\text{Pr}\left( \boldsymbol{x}_t \right)$ of the user’s current location through public tracking, check-in data set, or statistical information \cite{chatzikokolakis2015constructing}. Then, the attacker can calculate the posterior probability distribution $\mathbf{p}_{t}^{+}=\text{Pr}\left( \boldsymbol{x}_t|\boldsymbol{x}_{t}^{'} \right)$ after observing the user’s reported location $\boldsymbol{x}_{t}^{'}$, i.e.,
\begin{align}\label{Posterior probability}
\mathbf{p}_{t}^{+}=\text{Pr}\left( \boldsymbol{x}_t|\boldsymbol{x}_{t}^{'} \right) =\frac{\text{Pr}\left( \boldsymbol{x}_t \right) f\left( \boldsymbol{x}_{t}^{'}|\boldsymbol{x}_t \right)}{\sum_{\boldsymbol{x}_t\in \mathcal{A}}{\text{Pr}\left( \boldsymbol{x}_t \right) f\left( \boldsymbol{x}_{t}^{'}|\boldsymbol{x}_t \right)}}.
\end{align}

The {\textbf{optimal inference attack}} aims to infer the actual location at time $t$ by minimizing the expected inference error against the posterior distribution. Therefore, the inferred location $\boldsymbol{\hat{x}}_t$ is
\begin{align}\label{optimalinference}
\boldsymbol{\hat{x}}_t=\underset{\boldsymbol{\hat{x}}_t\in \mathcal{A}}{\text{arg}\min}\sum_{\boldsymbol{x}_t\in \mathcal{A}}{\text{Pr}\left( \boldsymbol{x}_t|\boldsymbol{x}_{t}^{'} \right)}d\left( \boldsymbol{\hat{x}}_t,\boldsymbol{x}_t \right).
\end{align}

\setlength{\textfloatsep}{10pt plus 1.0pt minus 2.0pt}
\begin{table}[!t]
	\centering
	\caption{List of Notations}
	\label{tab:1}
	\begin{tabular}{llllllllllll}
		\hline\hline\noalign{\smallskip}
		Symbol & Description  \\
		\noalign{\smallskip}\hline\noalign{\smallskip}
		$\boldsymbol{x}_i$ & Location $i$\\
		$\mathcal{A}$ & User map\\
		$Tr$ & User trajectory\\
		$Tr'$ & $\mathbf{Attacker}_{\mathcal{T}}$ inferred trajectory\\
		$\boldsymbol{x}_t/\boldsymbol{x}_{t}^{'}/\boldsymbol{\hat{x}}_t$ & Actual/perturbed/inferred location at time $t$\\
		$\mathbf{M}$ & Location transition probability matrix\\
		$\mathbf{p}_{t}^{-}/\mathbf{p}_{t}^{+}$ & Prior/Posterior probability at time $t$\\
		$L_i$ & Stay duration of $\boldsymbol{x}_i$\\
		$F_i$ & Access Frequency of $\boldsymbol{x}_i$\\
		$C_i$ & Semantic sensitivity of $\boldsymbol{x}_i$\\
		$S_i$ & Sensitivity of $\boldsymbol{x}_i$\\
		$G$ & Directed graph\\
		$B_i$ & Adjacent nodes set of $\boldsymbol{x}_i$\\
		$\epsilon _s$ & Total privacy budget of sensitive locations\\
		$E_m$ & Expected inference error bound\\
		$\delta$ & $0<\delta <1$\\
		$\varDelta \chi _t$ & Possible location set at time $t$\\
		$\Phi _t$ & Protection Location Set (PLS) at time $t$\\
		$D\left( \Phi _t \right)$ & Diameter of $\Phi _t$\\
		\noalign{\smallskip}\hline
	\end{tabular}
\end{table}

The function $d\left( \boldsymbol{\hat{x}}_t,\boldsymbol{x}_t \right)$ quantifies the discrepancy between the attacker's inferred location $\boldsymbol{\hat{x}}_t$ and the user's true location $\boldsymbol{x}_t$. The value of this distance directly reflects the severity of privacy leakage—a smaller distance indicates that the attacker's inference is closer to the actual location, representing a higher privacy risk. Specifically, when $d\left( \boldsymbol{\hat{x}}_t,\boldsymbol{x}_t \right)=0$, it signifies that the attacker has perfectly reconstructed the user's location, resulting in complete privacy leakage and a fully successful attack. The simulation results are shown in Figs. \ref{fig8a} and \ref{fig8b}. 

If $d\left( \boldsymbol{\hat{x}}_t,\boldsymbol{x}_t \right)$ represents the hamming distance, when $\boldsymbol{\hat{x}}_t=\boldsymbol{x}_t$, then $d\left( \boldsymbol{\hat{x}}_t,\boldsymbol{x}_t \right) =0$; otherwise, $d\left( \boldsymbol{\hat{x}}_t,\boldsymbol{x}_t \right) =1$. In this way, (\ref{optimalinference}) takes the location with the maximum posterior probability as the inferred location, called the {\textbf{Bayesian inference attack}, described as follows:
\begin{align}\label{Bayesianinference}
\boldsymbol{\hat{x}}_t=\underset{\boldsymbol{\hat{x}}_t\in \mathcal{A}}{\text{arg}\max}\,\,\text{Pr}\left( \boldsymbol{x}_t|\boldsymbol{x}_{t}^{'} \right).
\end{align}

The posterior probability $\Pr\left(\boldsymbol{x}_t|\boldsymbol{x}^{\prime}_t\right)$ represents the attacker's confidence in inferring the original location $\boldsymbol{x}_t$ after observing the perturbed location $\boldsymbol{x}^{\prime}_t$. Its magnitude directly correlates with the threat level of the attack: a higher probability indicates that the attacker can more accurately deduce the true location based on background knowledge, reflecting weaker defense effectiveness of the privacy-preserving mechanism. Conversely, a lower probability demonstrates that the privacy protection strategy successfully reduces inference accuracy. The simulation results are shown in Figs. \ref{fig8c} and \ref{fig8d}.

\textbf{Definition 1. ($\mathbf{Attacker}_{\mathcal{T}}$)} Based on the above, we define an attacker $\mathbf{Attacker}_{\mathcal{T}}$ who knows the spatiotemporal correlation information between locations, i.e., the user's location transition probability matrix $\mathbf{M}$. 
Within the domain of trajectory privacy research, inference attack constitutes a paramount security threat owing to their minimal data dependency and superior inferential power \cite{de2013unique}. The inherent regularity characterizing human mobility behaviors allows adversaries to accurately infer sensitive locations through the analysis of spatiotemporal correlations and intrinsic movement patterns.
Because $\mathbf{Attacker}_{\mathcal{T}}$ has the road network information between locations and the user's mobile profile, this kind of attacker can infer the user's actual location more accurately by excluding the untrusted perturbed locations. It can infer the user's location at the next moment according to the user's current location. By performing an inference attack on the location of each moment on the trajectory, $\mathbf{Attacker}_{\mathcal{T}}$ can infer the user's trajectory at consecutive moments and obtain the user's personal privacy information. The specific process is as follows:

$\mathbf{Attacker}_{\mathcal{T}}$ can infer the prior probability $\mathbf{p}_{t+1}^{-}$ of the user at time $t+1$, i.e.,
\begin{align}\label{t+1 prior probability}
\mathbf{p}_{t+1}^{-}=\mathbf{p}_{t}^{+}\mathbf{M}.
\end{align}
Using $\mathbf{p}_{t+1}^{-}$ $\mathbf{Attacker}_{\mathcal{T}}$ can identify with high confidence the locations where the user will likely not be at time $t+1$.
By deriving the posterior probability $\mathbf{p}_{t+1}^{+}$ at time $t+1$ according to (\ref{Posterior probability}), $\mathbf{Attacker}_{\mathcal{T}}$ can perform an optimal inference attack or Bayesian inference attack on the user's location at time $t+1$ according to (\ref{optimalinference}) or (\ref{Bayesianinference}) to obtain the inferred location $\boldsymbol{\hat{x}}_{t+1}$.
$\mathbf{Attacker}_{\mathcal{T}}$ conducts an attack at each moment along the user's trajectory to derive the inferred trajectory $Tr'=\left[ \boldsymbol{\hat{x}}_1,\boldsymbol{\hat{x}}_2,...,\boldsymbol{\hat{x}}_T \right]$.

\section{TRAJECTORY PRIVACY NOTIONS AND PROBLEM STATEMENT}\label{Trajectory Privacy Protection Statement}
In this section, we first list the primary concepts related to trajectory privacy and the criteria for determining PLS, and then we present this paper's problem statement.

\subsection{Directed Graph of Road Network}
Taking into account the location semantics, geo-topological relationship, and spatiotemporal accessibility between regions, the map $\mathcal{A}$ is transformed into a weighted directed graph $G=\left< V,E \right>$, where locations and roads are represented by the set of vertices $V$ and the set of edges $E$, respectively. Each edge and vertex are denoted by $e_{i,j}\in E$ and $\boldsymbol{x}_i\in V$ correspondingly, with $\boldsymbol{x}_i$ representing the vertex indexed as $i$ and $e_{i,j}$ linking $\boldsymbol{x}_i$ and $\boldsymbol{x}_j$. For example, in Fig. \ref{fig3}, $V=\left\{ \boldsymbol{x}_1,\boldsymbol{x}_2,\cdot \cdot \cdot ,\boldsymbol{x}_{13} \right\}$ and $E=\left\{ e_{1,2},e_{1,4},e_{2,1},\cdot \cdot \cdot ,e_{13,10} \right\}$. Notably, in a directed weighted graph, $e_{1,2}$ and $e_{2,1}$ are considered to be different edges. The weights associated with the edges indicate the Euclidean distance between the two regions, such as $d(\boldsymbol{x}_4, \boldsymbol{x}_9)$.

\setlength{\textfloatsep}{10pt plus 1.0pt minus 2.0pt}
\begin{figure}[!t]
\begin{center}
\includegraphics[width=0.4\textwidth]{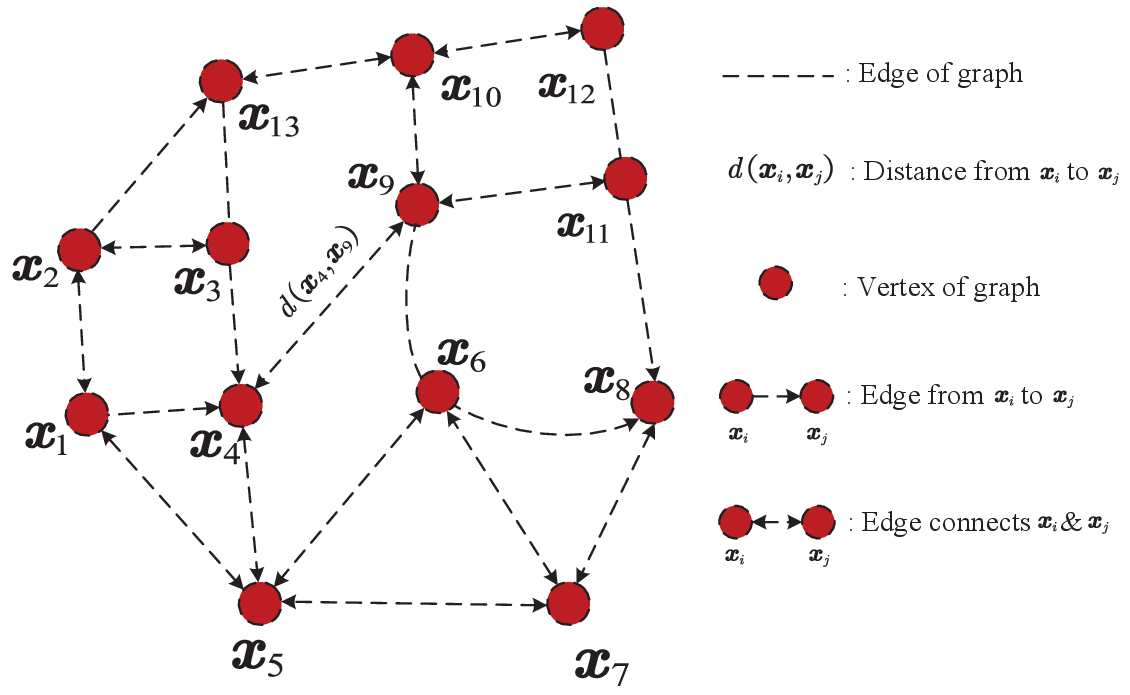}
\end{center}
\captionsetup{skip=-2pt}
\caption{Directed graph of the road network.}
\label{fig3}
\end{figure}
\subsection{Location Transition Probability Matrix}
The Matrix $\mathbf{N}$ is the location transfer matrix, representing the frequency of user transitions between different regions. Let $n_{ij}$ be an element in the $i$-th row and the $j$-th column of matrix $\mathbf{N}$, and $n_{ij}$ represents the number of times the user moves from region $\boldsymbol{x}_{i}$ to region $\boldsymbol{x}_{j}$.

Through the location transition matrix $\mathbf{N}$, the location transition probability matrix $\mathbf{M}$ of the user can be analyzed. Let $m_{ij}$ be an element in the $i$-th row and the $j$-th column of matrix $\mathbf{M}$, $m_{ij}=\dfrac{n_{ij}}{\sum_j{n_{ij}}}$ represents the probability of the user transitioning from $\boldsymbol{x}_{i}$ to $\boldsymbol{x}_{j}$. The matrix $\mathbf{M}$ describes the spatiotemporal correlation of the user across various locations within a trajectory.
\subsection{$\delta$-Location Set}
To protect locations frequently visited by the user, $\delta$-location set is introduced in \cite{xiao2015protecting}. This set, denoted as $\varDelta \chi _t$,  represents the locations where the user is most likely present at time $t$.

$\varDelta \chi _t$ denotes a set containing the minimum number of locations at time $t$ with a prior probability sum not less than $1-\delta $ ($0<\delta <1$), i.e.,
\begin{equation}\label{18}
\varDelta \chi _t=\min\left\{ \boldsymbol{x}_i|\sum_{\boldsymbol{x}_i}{\mathbf{p}_{t}^{-}\left[ i \right]}\ge 1-\delta \right\}.
\end{equation}

Given that the $\delta$-location set encompasses potential locations with a high likelihood of the user's presence at time $t$, the actual location $\boldsymbol{x}_t$ of the user may be eliminated with an extremely low probability. In this case, the nearest location $\boldsymbol{\tilde{x}}_t$ is substituted for the actual location $\boldsymbol{x}_t$, given by
\begin{equation}\label{draft location}
\boldsymbol{\tilde{x}}_t=\underset{\boldsymbol{\tilde{x}}_t\in \varDelta \chi _t}{\text{arg}\min}\ d\left( \boldsymbol{\tilde{x}}_t,\boldsymbol{x}_t \right).
\end{equation}
If $\boldsymbol{x}_t\in \varDelta \chi _t$, then $\boldsymbol{x}_t$ is protected within $\varDelta \chi _t$; otherwise, $\boldsymbol{\tilde{x}}_t$ is protected within $\varDelta \chi _t$.

\subsection{Condition for Determining PLS}

A two-phase dynamic differential location privacy framework known as PIVE has been proposed in \cite{yu2017dynamic}. This framework examines the complementary relationship between geo-indistinguishability and distortion privacy. It derives formulas to calculate the posterior probability's upper bound and the expected inference error's lower bound. By integrating these two privacy concepts, PIVE introduces a user-defined inference error bound, denoted as $E_m$, to determine PLS.

First, to guarantee the expected inference error in terms of PLS, the conditional expected inference error is given by
\begin{equation}\label{8}
ExpEr\left( \boldsymbol{x}_{t}^{'} \right) =\underset{\boldsymbol{\hat{x}}_t\in \mathcal{A}}{\min}\sum_{\boldsymbol{x}_t\in \mathcal{A}}{\text{Pr}\left( \boldsymbol{x}_t|\boldsymbol{x}_{t}^{'} \right)}d\left( \boldsymbol{\hat{x}}_t,\boldsymbol{x}_t \right).
\end{equation}
Given that the adversary narrows possible guesses to the PLS $\Phi _t$ that contains the user’s true location, we define
\begin{equation}\label{E}
E\left( \Phi _t \right) =\underset{\boldsymbol{\hat{x}}_t\in \mathcal{A}}{\min}\sum_{\boldsymbol{x}_t\in \Phi _t}{\frac{\text{Pr}\left( \boldsymbol{x}_t \right)}{\sum_{\boldsymbol{y}_t\in \Phi _t}{\text{Pr}\left( \boldsymbol{y}_t \right)}}}d\left( \boldsymbol{\hat{x}}_t,\boldsymbol{x}_t \right).
\end{equation}

Based on the lower bound of expected inference error,
\begin{equation}\label{15}
ExpEr\left( \boldsymbol{x}_{t}^{'} \right) \ge e^{-\epsilon}E\left( \Phi _t \right),
\end{equation}
the authors in \cite{yu2017dynamic} (Theorem 1) obtain a sufficient condition,
\begin{equation}\label{PLScondition}
E\left( \Phi _t \right) \ge e^{\epsilon}E_m,
\end{equation}
to satisfy the user-defined threshold, $\forall \boldsymbol{x}_{t}^{'}$, $ExpEr( \boldsymbol{x}_{t}^{'} ) \ge E_m$.

\subsection{Problem Statement}
The road network inherently constrains mobile users. Failing to take into account the user's specific road network exposes vulnerabilities that the attacker can exploit to collect extensive background information for malicious purposes. To enhance the protection of the user's current location,  it is essential to fully consider the structural characteristics of the road network. 
However, merely protecting the user's current location is insufficient due to the intricate spatiotemporal correlations between various locations along a trajectory.
Furthermore, each location along a trajectory varies in sensitivity due to differences in the user's stay duration, access frequency, and semantic sensitivity. Consequently, privacy requirements vary among individual users. Therefore, we develop a personalized trajectory privacy protection mechanism PTPPM, offering a better balance between trajectory privacy and QoS requirements.


\section{Personalized Trajectory Privacy Protection Mechanism} \label{Personalized Trajectory Privacy Protection Mechanism}

In this section, we propose a novel personalized trajectory privacy protection mechanism PTPPM, illustrated in Fig. \ref{fig4}.
Our approach integrates geo-indistinguishability and distortion privacy measures to defend against the attacker who knows the spatiotemporal correlations among various locations along the trajectory, thereby protecting the user's personalized trajectory privacy. Specifically, we first use algorithm $\mathcal{F}_1$ to obtain the set of possible locations for the user at each time instance, leveraging the associated prior probability at each moment along the trajectory.
We adopt algorithm $\mathcal{F}_2$ to allocate privacy budgets to individual locations based on their sensitivity and the underlying road network information.
Subsequently, algorithm $\mathcal{F}_3$ dynamically selects PLS for each possible trajectory location, incorporating geo-indistinguishability and distortion privacy. Notably, our mechanism facilitates personalized trajectory privacy protection through customizable privacy settings, such as expected inference error bound and privacy budget.
Finally, we put forth a novel Permute-and-Flip mechanism $\mathcal{K}$ to generate perturbed locations $\boldsymbol{x}_{t}^{'}$ within the PLS. These perturbed locations are strategically selected with a smaller perturbation distance, ensuring a better QoS experience while providing robust and effective privacy protection.

Allocating a global privacy budget can sometimes result in uniform protection across all locations, potentially leading to imbalances where highly sensitive areas receive inadequate protection while less critical ones are overly protected. To tackle this challenge, we present a novel approach that entails assigning privacy budgets to specific locations along the trajectory according to their individual sensitivity levels and road network constraints. Through this personalized privacy budget allocation approach, we aim to cater to the personalized privacy needs of each user.

\setlength{\textfloatsep}{10pt plus 1.0pt minus 2.0pt}
\begin{figure}[!t]
\begin{center}
\includegraphics[width=0.5\textwidth]{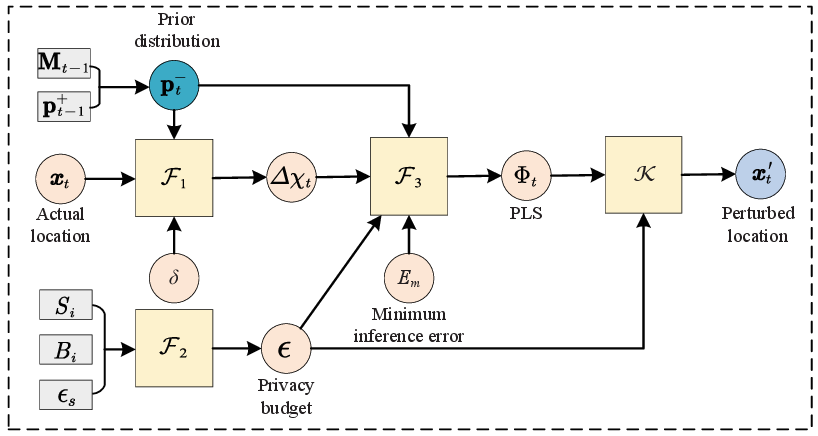}
\end{center}
\captionsetup{skip=-2pt}
\caption{The framework of PTPPM.}
\label{fig4}
\end{figure}

\subsection{Personalized Privacy Budget Allocation}
To better meet the personalized needs of users and defense against $\mathbf{Attacker}_{\mathcal{T}}$, we have developed a PPBA algorithm. This algorithm is designed carefully considering location sensitivity and road network constraints.

\subsubsection{Location Sensitivity}
Various locations hold varying degrees of sensitive information for the user. Factors like stay duration, access frequency, and semantic sensitivity contribute to the sensitivity levels associated with different locations.
\paragraph{Stay Duration}
By analyzing the user's mobile trajectory, it is discovered that each user's average daily activity time follows a notable power law distribution. Based on the user's historical trajectory data over a specific period, the average duration of the user's visit to each location is calculated, which can reflect the degree of the user's dependence on the location \cite{wen2020privacy}. The user's stay duration $L_i$ at location $\boldsymbol{x}_i$ is defined as
\begin{equation}
L_i=\frac{\int_{t_1}^{t_2}{f\left( t \right) dt}}{t_2-t_1},
\end{equation}
where $t_1$ is the start time of the access record and $t_2$ is the end time of the access record. $f(t) =1$ when the user is in the region and $f(t) =0$ when the user is not in the region.
\paragraph{Access Frequency}
Access frequency represents the ratio of user visits to a specific location over a defined period to the total number of visits to all locations along the trajectory.
Frequent visits to a location may indicate its significance to the user and suggest a higher likelihood of containing private user information \cite{wen2020privacy}. The user's access frequency $F_i$ at location $\boldsymbol{x}_i$ is
\begin{equation}
F_i=\frac{N\left( \boldsymbol{x}_i \right)}{\sum_{\boldsymbol{x}_i\in Tr}{N\left( \boldsymbol{x}_i \right)}},
\end{equation}
where $N\left( \boldsymbol{x}_i \right)$ is the number of user visits to location $\boldsymbol{x}_i$. $Tr=\left[ \boldsymbol{x}_1,\boldsymbol{x}_2,...,\boldsymbol{x}_T \right]$, $Tr$ represents the user's trajectory in continuous time length $T$.

\paragraph{Semantic Sensitivity}
The semantic sensitivity of a location, denoted as $C$, reflects the degree of sensitivity the user associates with a specific location. For instance, with $C\in \left\{ 1,2,3,4 \right\}$, a value of $C_i=1$ means that the location $\boldsymbol{x}_i$ is non-sensitive to the user (e.g., a park), while $C_i=4$ indicates high sensitivity of the location $\boldsymbol{x}_i$ to the user (e.g., a hospital) \cite{min2023indoor}.

Using the user's historical trajectory records over a specified period, metrics such as stay duration, access frequency, and semantic sensitivity are calculated for each location. These metrics are combined with weighted values to form a sensitivity level function. The user's sensitivity $S_i$ at location $\boldsymbol{x}_i$ is determined by

\begin{footnotesize}
\begin{eqnarray} \label{sensitive}
S_i&=&\alpha L_i+\beta F_i+\gamma C_i     \nonumber    \\
~&=&\alpha \frac{\int_{t_1}^{t_2}{f\left( t \right) dt}}{t_2-t_1}+\beta \frac{N\left( \boldsymbol{x}_i \right)}{\sum_{\boldsymbol{x}_i\in Tr}{N\left( \boldsymbol{x}_i \right)}}+\gamma C_i,
\end{eqnarray}
\end{footnotesize}where $\alpha$, $\beta$, and $\gamma$ represent the weighted values assigned to each sensitive factor in the calculation.

The weight parameters are dynamically configurable to adapt to different privacy requirements. For example, in sensitive scenarios like healthcare, the system increases the weight of semantic sensitivity $\gamma$, even with low stay duration and access frequency. For users concerned about behavioral privacy, the system increases the weight assigned to access frequency $\beta$ to enhance protection effectiveness. In more complex scenarios, a balanced strategy is applied to address multi-dimensional privacy risks. This flexible configuration highlights the model’s adaptability to diverse privacy needs.

\setlength{\textfloatsep}{10pt plus 1.0pt minus 2.0pt}
{\setstretch{0.95}  
\begin{algorithm}[!t]
    \SetKwData{Left}{left}
    \SetKwData{This}{this}
    \SetKwData{Up}{up}
    \SetKwFunction{Union}{Union}
    \SetKwFunction{FindCompress}{FindCompress}
    \SetKwInOut{Input}{Input}
    \SetKwInOut{Output}{Output}

    \Input{$Tr=\left[ \boldsymbol{x}_1,\boldsymbol{x}_2,...,\boldsymbol{x}_T \right]$, $G=\left< V,E \right>$, $\epsilon _s$}
    \Output{Privacy budget allocation of sensitive locations and their adjacent nodes.}
    \BlankLine

    \For{sensitive location $\boldsymbol{x}_i$, $i=1$ \KwTo $n$}
    {
        Determine $L_i$, $F_i$, $C_i$, and $B_i$\;
        Calculate $S_i=\alpha L_i+\beta F_i+\gamma C_i$\;
        Privacy budget allocation of $\epsilon _{\boldsymbol{x}_i}$ via (\ref{PrivacyBudgetSi})\;
        \For{$\boldsymbol{x}_j\in B_i$, $j=1$ \KwTo $m$}
            {\label{forins}
            Calculate $d\left( \boldsymbol{x}_j,\boldsymbol{x}_i \right)$\;
            Privacy budget allocation of $\epsilon _{\boldsymbol{x}_j}^{'}$ via (\ref{PrivacyBudgetNode})\;
            \If{$\boldsymbol{x}_j$ is also a sensitive location}
            {\label{lt} The privacy budget of $\boldsymbol{x}_j$ is $\min \left( \epsilon _{\boldsymbol{x}_j},\epsilon _{\boldsymbol{x}_j}^{'} \right)$;}

     }
             }
            \caption{Personalized Privacy Budget Allocation Algorithm (PPBA)} \label{Algorithm1}
    \end{algorithm}
}

\subsubsection{Privacy Budget Allocation of Sensitive Locations}
Users exhibit diverse privacy protection needs across various sensitive locations. To prevent reducing the accuracy and QoS of the perturbed trajectory, we implement distinct levels of perturbation tailored to each location's sensitivity. Through PPBA, we customize the allocation of privacy budgets according to the specific location sensitivities, ensuring a personalized and effective privacy protection strategy.

Based on historical trajectory data, the PPBA computes the user's stay duration, access frequency, and semantic sensitivity toward all sensitive locations within a specified period. Subsequently, the sensitivity of each sensitive location is determined by (\ref{sensitive}). In accordance with these sensitivity calculations, the algorithm assigns the appropriate privacy budget to the sensitive locations along the trajectory. The allocation formula is outlined as follows:
\begin{equation}\label{PrivacyBudgetSi}
\epsilon _{\boldsymbol{x}_i}=\frac{\frac{1}{S_i}}{\sum_{\boldsymbol{x}_i\in \mathcal{A}}{\frac{1}{S_i}}}\cdot \epsilon _s,
\end{equation}
where $\epsilon _{\boldsymbol{x}_i}$  denotes the privacy budget assigned to the sensitive location $\boldsymbol{x}_i$, and $S_i$ is the sensitivity associated with this specific location. The parameter $\epsilon _s$ represents the total privacy budget allocated for all sensitive locations.
\subsubsection{Privacy Budget Allocation of Adjacent Nodes}
In our approach, it is imperative not only to secure the sensitive location itself but also to appropriately perturb the surrounding locations. It is essential to prevent $\mathbf{Attacker}_{\mathcal{T}}$ from exploiting spatiotemporal correlations among locations to infer sensitive locations.

The PPBA identifies a set $B_i$ of adjacent nodes surrounding a sensitive location $\boldsymbol{x}_i$ in a directed graph $G$ representing the road network. PPBA allocates a privacy budget to nodes within $B_i$ based on the distance from the sensitive location to its neighboring nodes. The allocation formula is defined as:
\begin{equation}\label{PrivacyBudgetNode}
\epsilon _{\boldsymbol{x}_j}^{'}=\frac{\sum_{\boldsymbol{x}_j\in B_i}{\frac{1}{d\left( \boldsymbol{x}_i,\boldsymbol{x}_j \right)}}}{\frac{1}{d\left( \boldsymbol{x}_i,\boldsymbol{x}_j \right)}}\cdot \epsilon _{\boldsymbol{x}_i},
\end{equation}
where $\epsilon _{\boldsymbol{x}_j}^{'}$ represents the privacy budget allocated to node $\boldsymbol{x}_j$, and $d\left( \boldsymbol{x}_i,\boldsymbol{x}_j \right)$ represents the Euclidean distance between the sensitive location $\boldsymbol{x}_i$ and node $\boldsymbol{x}_j$. If node $\boldsymbol{x}_j$ is also a sensitive location, $\min  ( \epsilon _{\boldsymbol{x}_j},\epsilon _{\boldsymbol{x}_j}^{'})$ is selected as the privacy budget of that node.

As illustrated in Algorithm  \ref{Algorithm1}, we summarize the process of the PPBA algorithm.
The calculation of the above parameters depends on the user’s mobile profile, i.e., the attackers' prior knowledge, which can be updated over time as the dataset is updated. The proposed method can recalculate the parameters accordingly to determine the allocation of personalized privacy budgets, ensuring both adaptability and effectiveness in defending against inference attacks.

The PPBA algorithm offers a space complexity of $O(n)$ and enables efficient processing with a time complexity of $O(nm)$, where $n$ is the number of sensitive locations and $m$ denotes the number of adjacent nodes per location. Its complexity is mainly determined by the underlying directed graph that models the road network. To ensure computational efficiency, we employ a standard adjacency list to store and manage the graph structure. These low computational requirements guarantee scalability even in large-scale urban networks.

\subsection{Determine $\varDelta \chi_t$ at Continuous  Times}
The transition probability matrix $\mathbf{M}$ is constructed according to the user's historical trajectory data and behavior habits \cite{xiao2015protecting}.
We eliminate all impossible locations ($\mathbf{p}_{t}^{-}$ is minimal or $\mathbf{p}_{t}^{-}=0$) based on certain criteria to obtain the possible location set at time $t$, i.e., $\varDelta \chi _t$.
If the actual location at time $t$ is removed, we substitute it with $\boldsymbol{\tilde{x}}_t$.

We calculate the posterior probability $\mathbf{p}_{t}^{+}$ according to (\ref{Posterior probability}) and then combine the location transition probability matrix $\mathbf{M}$ according to (\ref{t+1 prior probability}) to obtain the prior probability $\mathbf{p}_{t+1}^{-}$ at time $t+1$. In terms of $\mathbf{p}_{t+1}^{-}$, we get $\varDelta \chi _{t+1}$ at time $t+1$. The size of $\varDelta \chi _{t+1}$ is determined by the value of $\delta$. Then, we obtain $\varDelta \chi _t$ at consecutive times by following the same process.
\subsection{Determine Protection Location Set}
Once we have obtained $\varDelta \chi _t$ for each timestamp along the trajectory, our next focus is determining the PLS $\Phi_t$ within $\varDelta \chi _t$ for any specific location.

In order to improve the user's QoS, the smaller the diameter $D\left(\Phi_t \right)$ of the circular area, the better. Since $D\left( \Phi _t \right)$ is the diameter of the $\Phi _t$, the distance between any two locations is less than or equal to $D\left( \Phi _t \right)$. For $\forall \boldsymbol{x}_t,\boldsymbol{\hat{x}}_t$ in  $\Phi _t$, we have $D\left( \Phi _t \right) \ge d\left( \boldsymbol{x}_t,\boldsymbol{\hat{x}}_t \right)$. By (\ref{PLScondition}), we have

\begin{small}
\begin{equation}\label{24}
e^{\epsilon}E_m\le E\left( \Phi _t \right) \le \underset{\boldsymbol{\hat{x}}_t\in \Phi _t}{\min}\sum_{\boldsymbol{x}_t\in \Phi _t}{\frac{\text{Pr}\left( \boldsymbol{x}_t \right)}{\sum_{\boldsymbol{y}_t\in \Phi _t}{\text{Pr}\left( \boldsymbol{y}_{\boldsymbol{t}} \right)}}}D\left( \Phi _t \right) =D\left( \Phi _t \right).
\end{equation}
\end{small}

We assign the privacy budget $\epsilon$ individually by the PPBA algorithm. $E_m$ is dynamically adjusted to determine the conditions that PLS should satisfy.
To effectively find the PLS with the smallest diameter at time $t$, we employ the minimum distance search method based on the Hilbert curve outlined in \cite{yu2017dynamic}. For each possible location $\boldsymbol{x}_t$ in $\varDelta \chi _t$ along the trajectory, we search the neighborhood of $\boldsymbol{x}_t$ following the search direction of the Hilbert curve.
By adhering to the conditions specified in (\ref{PLScondition}), we identify the PLS for $\boldsymbol{x}_t$  and select the one with the smallest diameter as the PLS $\Phi _t$.

On this basis, to mitigate the risk of an oversized protection area due to the single-direction search of the Hilbert curve, we introduce spatial rotation to enhance the likelihood of discovering a PLS for each location $\boldsymbol{x}_t$ with a reduced diameter.
Specifically, similar to \cite{yu2017dynamic}, we rotate the curve clockwise by 90, 180, and 270 degrees around the center point to generate three additional Hilbert curves. After rotation, we search the PLS under each Hilbert curve with the user's location and select the group with the smallest diameter from the four outcomes as the PLS.

\setlength{\textfloatsep}{10pt plus 1.0pt minus 2.0pt}
{\setstretch{0.95}  
\begin{algorithm}[!t]
    \SetKwData{Left}{left}
    \SetKwData{This}{this}
    \SetKwData{Up}{up}
    \SetKwFunction{Union}{Union}
    \SetKwFunction{FindCompress}{FindCompress}
    \SetKwInOut{Input}{Input}
    \SetKwInOut{Output}{Output}

    \Input{$\mathbf{p}_{t-1}^{+}$, $\mathbf{M}$, $\delta$, $Tr=\left[ x_1,x_2,...,x_T \right]$, $G=\left< V,E \right>$, $\epsilon _s$}
    \Output{$\mathbf{p}_{t}^{+}$, $\boldsymbol{x}_{t}^{'}$}
    \BlankLine
    $\mathbf{p}_{t}^{-}\gets \mathbf{p}_{t-1}^{+}\mathbf{M}$\;
    Determine $\varDelta \chi _t$ from $\mathbf{p}_{t}^{-}$, $\delta$\;
    \If{$\boldsymbol{x}_t\notin \varDelta \chi _t$}
            {\label{lt} $\boldsymbol{x}_t=\boldsymbol{\tilde{x}}_t$ via (\ref{draft location});}
    $\epsilon$ $\gets$ $\mathbf{Algorithm\ 1}$($G$, $\epsilon _s$, $Tr$)\;
    Dynamic adjustment of $E_m$\;
    \For{$\boldsymbol{x}_i\in \varDelta \chi _t$, $i=1$ \KwTo size of $\varDelta \chi _t$}
    {
        Find PLS $\Phi _t$ that satisfies (\ref{PLScondition}) via Hilbert curve \cite{yu2017dynamic}\;
        Select the smallest diameter $\Phi _t$ as PLS of $\boldsymbol{x}_t$\;
    }
    Release perturbed locations $\boldsymbol{x}_{t}^{'}$ by PF mechanism\;
    Calculate $\mathbf{p}_{t}^{+}$ via (\ref{Posterior probability})\;
    Go to the next timestamp;
            \caption{Personalized Trajectory Privacy Protection Mechanism (PTPPM)}\label{Algorithm2}
    \end{algorithm}
}

\subsection{Differentially Private Mechanism in Protection Location Set}
We put forth a new perturbation mechanism, Permute-and-Flip, to release the perturbed location with a smaller perturbation distance, which can better balance location privacy and QoS. Initially developed to protect privacy during data publication \cite{mckenna2020permute}, the PF mechanism is now used for the first time to protect the location within the PLS $\Phi _t$, exploiting the mapping correlation between the utility function and the Euclidean distance.
The Permute-and-Flip mechanism selects the query option with the highest score during query processing.
To balance privacy protection while minimizing QoS loss, we aim to output the location within the PLS closest to the actual location. However, when the output is the actual location, the user's privacy cannot be guaranteed, and it does not make sense to reduce the QoS loss.
This mechanism dynamically optimizes the selection of the PLS by expanding the search range, i.e., increasing the PLS diameter $D$, in response to location sparsity, influenced by (\ref{PLScondition}) and (\ref{24}). It demonstrates that this mechanism has strong adaptability in both typical urban and suburban scenarios.
To address this, we define the query function as the difference between the negative of the distance and the maximum non-zero negative distance value. The sensitivity of the utility function is defined as follows:
\begin{equation}\label{utility}
\begin{split}
\Delta u=&\mathop {\max  \max} \limits_{\boldsymbol{x}_{t}^{^{\prime}}\in \mathcal{A} ,\boldsymbol{x}_t,\boldsymbol{y}_t\in \Phi _t}\left| -d\left( \boldsymbol{x}_t,\boldsymbol{x}_{t}^{^{\prime}} \right) -\left( -d\left( \boldsymbol{y}_t,\boldsymbol{x}_{t}^{^{\prime}} \right) \right) \right|\\
&=\mathop {\max   \max} \limits_{\boldsymbol{x}_{t}^{^{\prime}}\in \mathcal{A} ,\boldsymbol{x}_t,\boldsymbol{y}_t\in \Phi _t}\left| d\left( \boldsymbol{y}_t,\boldsymbol{x}_{t}^{^{\prime}} \right) -d\left( \boldsymbol{x}_t,\boldsymbol{x}_{t}^{^{\prime}} \right) \right|,
\end{split}
\end{equation}
according to the triangle inequality, we have $
d( \boldsymbol{y}_t,\boldsymbol{x}_{t}^{^{\prime}} ) -d( \boldsymbol{x}_t,\boldsymbol{x}_{t}^{^{\prime}} ) \leqslant d( \boldsymbol{y}_t,\boldsymbol{x}_t ) \leqslant D(\Phi _t ).
$\par
After obtaining $\Delta \chi _t$ for each location on the trajectory, we can determine the corresponding $\Phi _t$ for each possible location within $\Delta \chi _t$ using (\ref{PLScondition}). Given the current location $\boldsymbol{x}_t$ and the PLS $\Phi_t$, the probability of the output perturbed location $\boldsymbol{x}_{t}^{'}$ is proportional to $
\exp \left( \frac{\epsilon (u(D,r)-\max(u(D,r)))}{2\Delta u} \right)$ according to the PF mechanism, where
$u( D,r ) =-d( \boldsymbol{x}_t,\boldsymbol{x}_{t}^{'} )$. In addition, our PLS contains actual location, in order to protect the user's privacy while balancing QoS, we prefer to output locations within the PLS that are close to actual locations other than itself. Therefore, $
\max ( u( D,r ) ) =\underset{\boldsymbol{x}_t,\boldsymbol{x}_{t}^{\prime}\in \Phi _t ,\boldsymbol{x}_t\ne \boldsymbol{x}_{t}^{\prime}}{\max}( -d( \boldsymbol{x}_t,\boldsymbol{x}_{t}^{\prime} ) ) = -d_{sm}( \boldsymbol{x}_t,\boldsymbol{x}_{t}^{^{\prime}} ),
$ where $d_{sm}( \boldsymbol{x}_t,\boldsymbol{x}_{t}^{^{\prime}})$ is the second smallest distance between actual and perturbed locations.

Thus, we have the perturbed locations’ probability distribution
\begin{footnotesize}
\begin{equation}
\begin{split}
f\left( \boldsymbol{x}_{t}^{^{\prime}}|\boldsymbol{x}_t \right) =\omega\exp \left( \frac{\epsilon \left( -d\left( \boldsymbol{x}_t,\boldsymbol{x}_{t}^{^{\prime}} \right) -\max \left( -d\left( \boldsymbol{x}_t,\boldsymbol{x}_{t}^{^{\prime}} \right) \right) \right)}{2D\left( \Phi _t \right)} \right)
\\
=\omega\exp \left( \frac{-\epsilon \left( d\left( \boldsymbol{x}_t,\boldsymbol{x}_{t}^{^{\prime}} \right) -d_{sm}\left( \boldsymbol{x}_t,\boldsymbol{x}_{t}^{^{\prime}} \right) \right)}{2D\left( \Phi _t \right)} \right) ,
\end{split}
\end{equation}
\end{footnotesize}where $\omega$ is the probability distribution normalization factor, i.e.,
\vspace{-5pt}
\begin{footnotesize}
\begin{equation}
\begin{split}
\omega&=\left( \sum_{\boldsymbol{x}_{t}^{\prime}\in \mathcal{A}}{\exp}\left( \frac{\epsilon \left( -d\left( \boldsymbol{x}_t,\boldsymbol{x}_{t}^{\prime} \right) -\max \left( -d\left( \boldsymbol{x}_t,\boldsymbol{x}_{t}^{\prime} \right) \right) \right)}{2D\left( \Phi _t \right)} \right) \right) ^{-1}
\\
&=\left( \sum_{\boldsymbol{x}_{t}^{\prime}\in \mathcal{A}}{\exp}\left( \frac{-\epsilon \left( d\left( \boldsymbol{x}_t,\boldsymbol{x}_{t}^{\prime} \right) -d_{sm}\left( \boldsymbol{x}_t,\boldsymbol{x}_{t}^{\prime} \right) \right)}{2D\left( \Phi _t \right)} \right) \right) ^{-1}.
\end{split}
\end{equation}
\end{footnotesize}The normalization factor distributes the values of the probabilities all within 0-1.

Algorithm \ref{Algorithm2} summarizes the overall process of the proposed PTPPM.

\subsection{Security Analysis}
In traditional privacy-preserving models, trajectory data is typically perturbed and anonymized on the server side, assuming full trust in the server. In contrast, the PTPPM performs perturbation locally on the user side before data upload, reducing reliance on server trust and ensuring that raw data remains on the user's device, thus mitigating risks of eavesdropping and malicious servers during transmission. The proposed PTPPM mechanism adopts a differential privacy model, where each released location independently satisfies $2\epsilon$-DP. Even with full auxiliary knowledge, an adversary cannot reliably infer the user's true location. Leveraging the sequential composition property, the mechanism maintains an overall $2\epsilon$-DP guarantee over time, effectively preventing privacy leakage from long-term data accumulation.

\begin{theorem}
The Permute-and-Flip mechanism satisfies $2\epsilon$-differential privacy on the PLS $\Phi$.
\end{theorem}
\begin{proof}
Starting from (\ref{utility}), we remove the absolute value and simplify the expression as follows:
\begin{equation}
d\left( \boldsymbol{x}_t,\boldsymbol{x}_{t}^{'} \right) -D\left( \Phi _t \right) \le d\left( \boldsymbol{y}_t,\boldsymbol{x}_{t}^{'} \right) \le d\left( \boldsymbol{x}_t,\boldsymbol{x}_{t}^{'} \right) +D\left( \Phi _t \right).
\end{equation}

Based on the inequality obtained above, we proceed to analyze the ratio of the conditional probabilities as follows:
\small
\begin{flalign}\label{22}
&\frac{f(x_t'|x_t)}{f(x_t'|y_t)}
= \frac{
w_{x_t} \cdot \exp\left( -\frac{\epsilon}{2D(\Phi_t)} \left[ d(x_t, x_t') - d_{sm}(x_t, x_t') \right] \right)
}{
w_{y_t} \cdot \exp\left( -\frac{\epsilon}{2D(\Phi_t)} \left[ d(y_t, x_t') - d_{sm}(y_t, x_t') \right] \right)
} \nonumber \\
&\le \frac{w_{x_t}}{w_{y_t}}
\exp\left( \frac{\epsilon \left( \left| d(x_t, x_t') - d(y_t, x_t') \right| + \left| d(x_t, x_t') - d(y_t, x_t') \right| \right)}{2D(\Phi_t)} \right) \nonumber \\
&\le \frac{w_{x_t}}{w_{y_t}} 
\exp\left( \frac{2\epsilon d(x_t, y_t)}{2D(\Phi_t)} \right)
\le \frac{w_{x_t}}{w_{y_t}}
\exp\left( \frac{\epsilon D(\Phi_t)}{D(\Phi_t)} \right) \nonumber \\
&\leq e^{\epsilon}\cdot \frac{\sum_{x_{t}^{'}\in \Delta \chi _t}{\exp}\left( -\frac{\epsilon \left[ d\left( x_t,x_{t}^{'} \right) -D\left( \Phi _t \right) -d_{sm}\left( x_t,x_{t}^{'} \right) -D\left( \Phi _t \right) \right]}{2D\left( \Phi _t \right)} \right)}{\sum_{x_{t}^{'}\in \Delta \chi _t}{\exp}\left( -\frac{\epsilon}{2D\left( \Phi _t \right)}\left[ d\left( x_t,x_{t}^{'} \right) -d_{sm}\left( x_t,x_{t}^{'} \right) \right] \right)} \nonumber \\
&\le e^{\epsilon} \cdot e^{\epsilon} \cdot
\frac{
\sum\limits_{x_t' \in \Delta \chi_t}
\exp\left( -\frac{\epsilon}{2D(\Phi_t)} \left[ d(x_t, x_t') - d_{sm}(x_t, x_t') \right] \right)
}{
\sum\limits_{x_t' \in \Delta \chi_t}
\exp\left( -\frac{\epsilon}{2D(\Phi_t)} \left[ d(x_t, x_t') - d_{sm}(x_t, x_t') \right] \right)
} \nonumber \\
&= e^{2\epsilon}
\end{flalign}
\end{proof}
Therefore, the PF mechanism satisfies $2\epsilon$-differential privacy.
\begin{remark}
The above proof shows that the perturbation mechanism satisfies \(2\epsilon\)-differential privacy. In (\ref{22}), \(d_{sm}(x_t, x'_t)\) denotes the second smallest distance between the actual and perturbed locations, reflecting a preference for selecting nearby alternatives. This ensures strong privacy protection while preserving QoS. 
\end{remark}

\begin{theorem}
If each sensitive point and its adjacent points in the trajectory are perturbed using the PF mechanism individually, and each mechanism satisfies $2\epsilon_{x_i}$-differential privacy, then the entire trajectory satisfies $2\epsilon_s$-differential privacy.
\end{theorem}

\begin{proof}
Each location point $x_i$ is independently perturbed using the PF mechanism, ensuring that it satisfies $2\epsilon_{x_i}$-differential privacy.

Since perturbations are independent, the total privacy loss of the trajectory accumulates according to the sequential composition theorem \cite{dwork2006differential}:
\begin{equation}
\sum_{x\in Tr}{2}\epsilon _{\boldsymbol{x}_j}^{'}=2\sum_{x\in Tr}{\epsilon _{x_i}}.
\end{equation}
According to (\ref{PrivacyBudgetSi}), the total privacy budget over all sensitive locations equals $\epsilon_s$. Since each adjacent location $x_j \in B_i$ shares the budget of its associated sensitive location 
 $x_i$, it holds that
\begin{equation}
\sum_{x \in B_i} \epsilon _{\boldsymbol{x}_j}^{'} \leq \epsilon_s.
\end{equation}
Thus, the total privacy loss is bounded by:
\begin{equation}
2\sum_{x\in Tr}\epsilon_x\leq2\epsilon_s.
\end{equation}
Hence, the PTPPM mechanism satisfies $2\epsilon_s$-differential privacy.

\end{proof}

\begin{remark}
At each time step, the released location independently satisfies \(2\epsilon_s\)-differential privacy. By the sequential composition property, users consistently receive location services with a privacy guarantee of \(2\epsilon_s\).
\end{remark}

\setlength{\textfloatsep}{8pt plus 1.0pt minus 2.0pt}
\begin{figure}[!t]
\begin{center}
\includegraphics[width=0.42\textwidth]{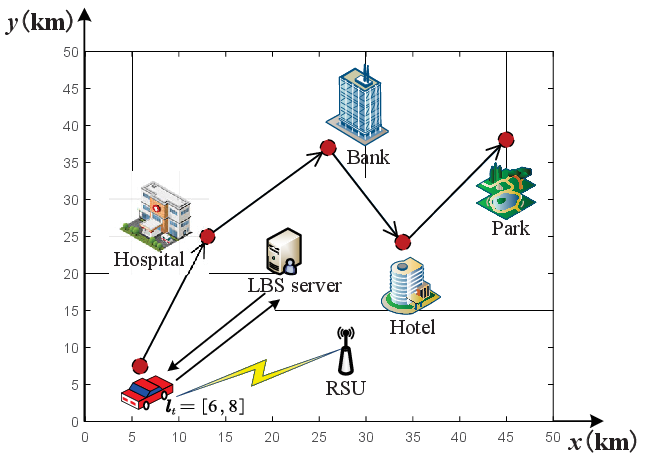}
\end{center}
\captionsetup{skip=-2pt}
\caption{Simulation setting of the trajectory of a user.}
\label{fig5}
\end{figure}

\setlength{\textfloatsep}{10pt plus 1.0pt minus 2.0pt}
\begin{figure*}[!t]
\centering
\subfigure[Trajectory privacy v.s. $\epsilon_s$ (T-Drive)]{\includegraphics[width=0.23\textwidth]{{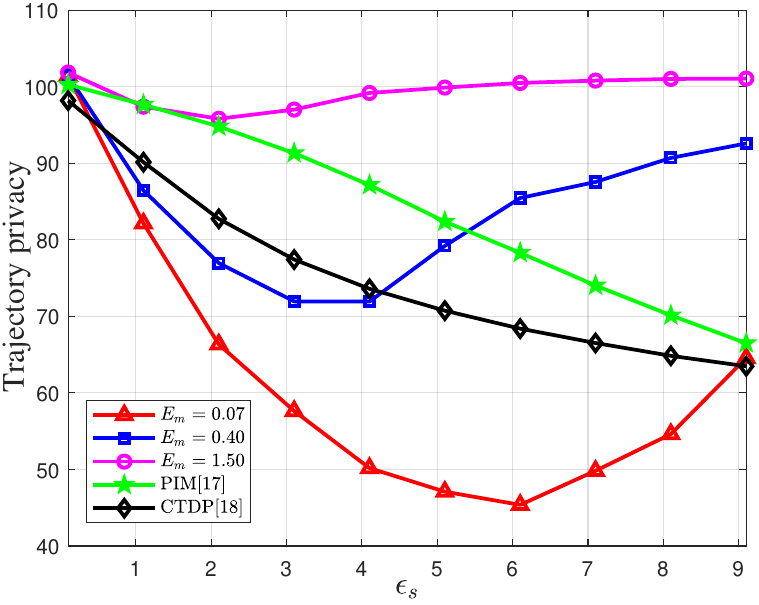}}%
\label{fig9a}}
\hfil
\subfigure[QoS loss v.s. $\epsilon_s$ (T-Drive)]{\includegraphics[width=0.23\textwidth]{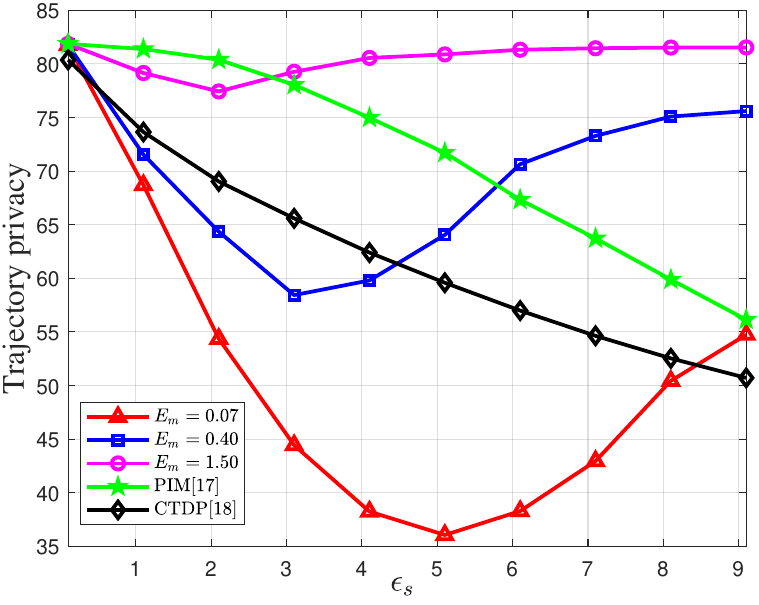}%
\label{fig9c}}
\hfil
\subfigure[Trajectory privacy v.s. $\epsilon_s$ (Geolife)]{\includegraphics[width=0.23\textwidth]{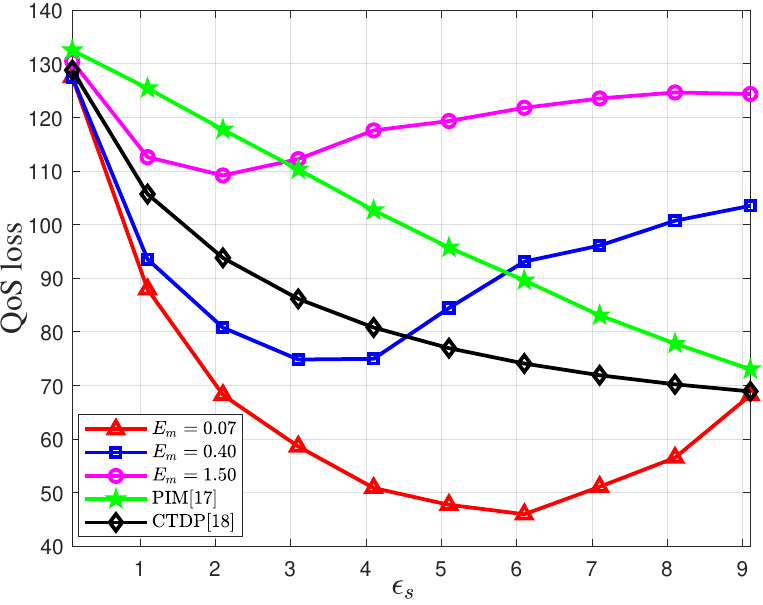}%
\label{fig9b}}
\hfil
\subfigure[QoS loss v.s. $\epsilon_s$ (Geolife)]{\includegraphics[width=0.23\textwidth]{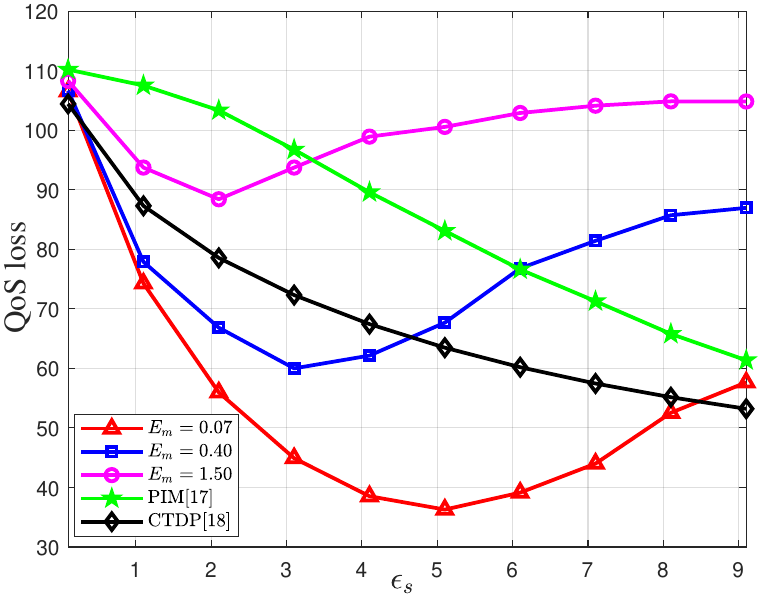}%
\label{fig9d}}
\hfil
\caption{Performance comparison of different TPPMs under varying $\epsilon _s$ on T-Drive and Geolife datasets.}
\label{fig9}
\end{figure*}

\section{Simulation Results}\label{Simulation Results}
\subsection{Simulation Setup}

In this section, we evaluate the impact of privacy parameters on our proposed mechanism, PTPPM, and conduct a performance analysis against different attacks. We compare the efficacy of PTPPM with other trajectory privacy protection mechanisms (TPPMs), including PIVE \cite{yu2017dynamic}, PIM \cite{xiao2015protecting}, and CTDP \cite{10666272} regarding location privacy and QoS loss.

To verify the effectiveness of the proposed mechanism and enhance its generalizability and adaptability across different application scenarios, this study employs both the T-Driver dataset \cite{yuan2010t}, \cite{yuan2011driving}, and the Geolife dataset \cite{zheng2009mining} for experiments. The T-Driver dataset contains GPS trajectory data collected from 10,357 taxis in Beijing from February 2 to February 8, 2008. It includes detailed information such as taxi ID, timestamp, and geographical coordinates (longitude and latitude). Each file in the dataset is named after the corresponding taxi ID and contains the complete trajectory data for that vehicle. The Geolife dataset, collected by Microsoft Research Asia, comprises a large number of real user trajectories, covering walking, biking, and a small portion of driving traces, thus providing more diverse mobility patterns and trajectory characteristics.

A file is randomly chosen for simulation purposes, and a subset of locations within the file is designated as sensitive locations. The stay duration, access frequency, and other relevant user information about these locations are then analyzed. A location transition probability matrix is constructed using all of the trajectory data from this selected file. The map is divided into cells measuring 0.62 × 0.62 km$^2$, and a time interval of 177 seconds is set. Fig. \ref{fig5} depicts the trajectory of this user at five sequential moments.

The evaluation of location privacy $p$ and QoS loss $q$ is conducted using metrics similar to those outlined in our previous work \cite{min20213d}, as follows:
\begin{equation}\label{Privacy}
p=\sum_{\boldsymbol{x}_t,\boldsymbol{x}_{t}^{'},\boldsymbol{\hat{x}}_t\in \mathcal{A}}{\text{Pr}\left( \boldsymbol{x}_t \right)}f\left( \boldsymbol{x}_{t}^{'}|\boldsymbol{x}_t \right) h\left( \boldsymbol{\hat{x}}_t|\boldsymbol{x}_{t}^{'} \right) d\left( \boldsymbol{x}_t,\boldsymbol{\hat{x}}_t \right),
\end{equation}
\begin{equation}\label{QoS loss}
q=\sum_{\boldsymbol{x}_t,\boldsymbol{x}_{t}^{'}\in \mathcal{A}}{\text{Pr}\left( \boldsymbol{x}_t \right)}f\left( \boldsymbol{x}_{t}^{'}|\boldsymbol{x}_t \right) d\left( \boldsymbol{x}_t,\boldsymbol{x}_{t}^{'} \right).
\end{equation}

\setlength{\textfloatsep}{10pt plus 1.0pt minus 2.0pt}
\begin{figure*}[!t]
\centering
\subfigure[Trajectory privacy v.s. $\epsilon_s$]{
\includegraphics[width=0.23\textwidth]{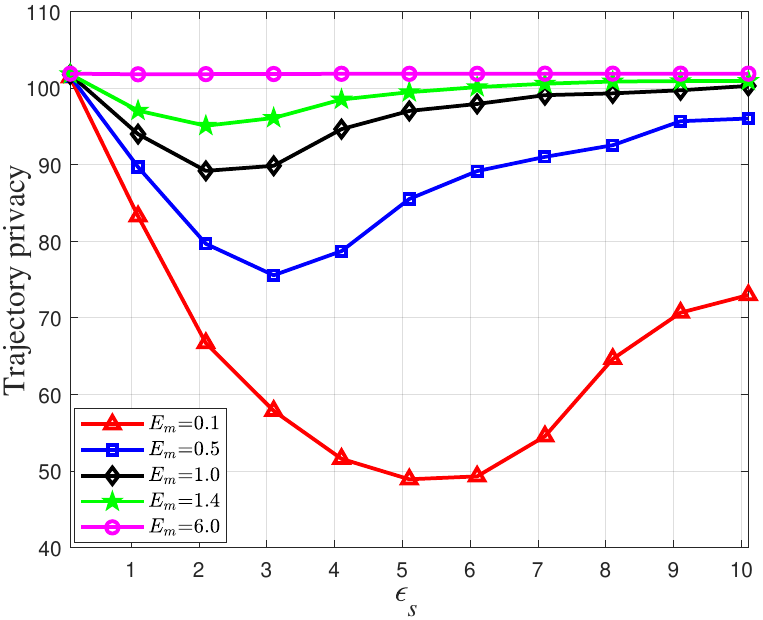} \label{fig7a}
}
\hfill
\subfigure[Trajectory privacy v.s. $E_m$]{
\includegraphics[width=0.23\textwidth]{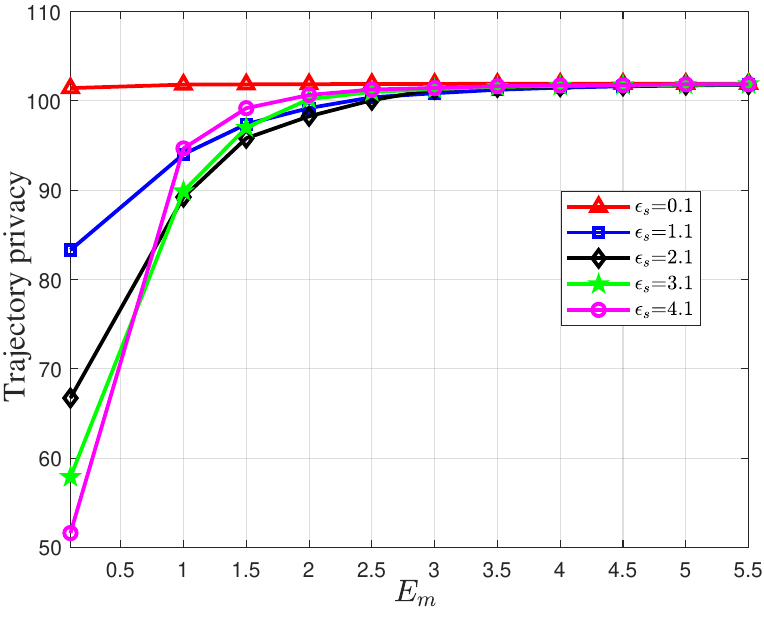} \label{fig7b}
}
\hfill
\subfigure[QoS loss v.s. $\epsilon_s$]{
\includegraphics[width=0.23\textwidth]{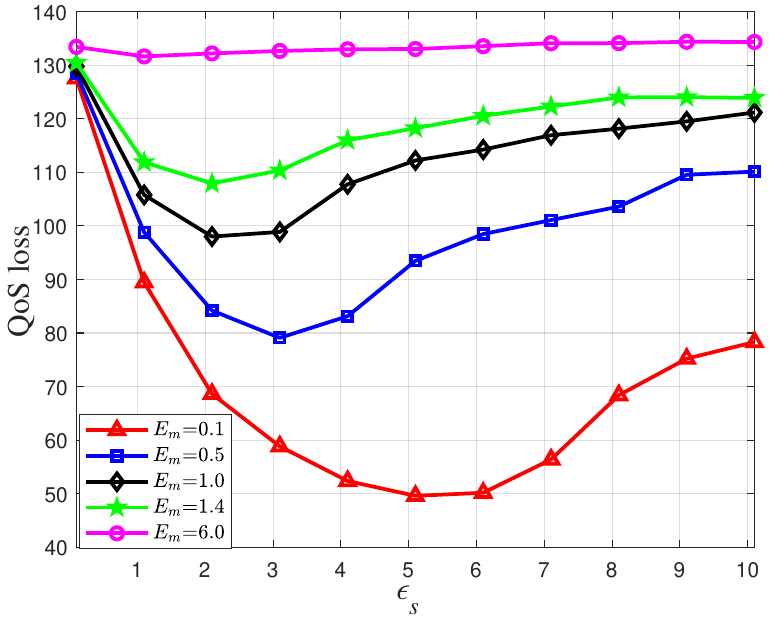}\label{fig7c}
}
\hfill
\subfigure[QoS loss v.s. $E_m$]{
\includegraphics[width=0.23\textwidth]{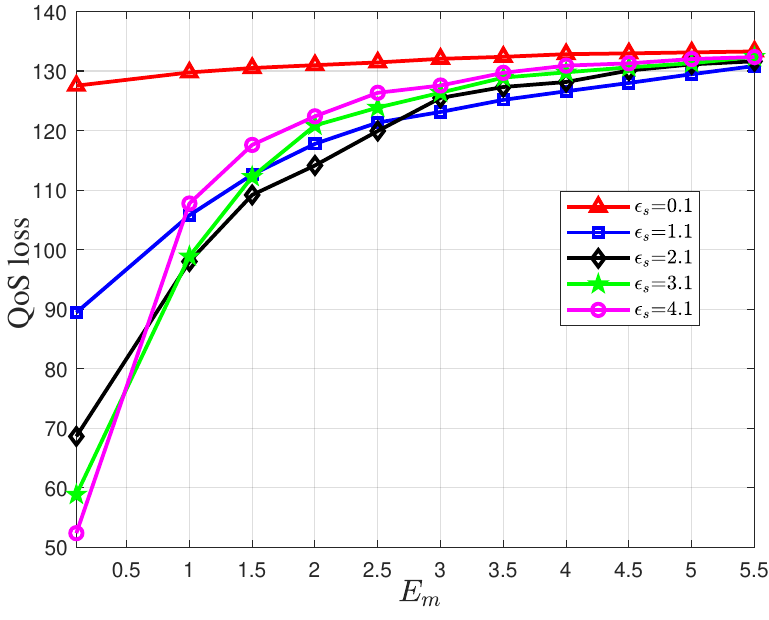}\label{fig7d}
}
\hfill
\captionsetup[subfigure]{skip=2pt}
\caption{Impact of $\epsilon$ and $E_m$ on personalized trajectory privacy protection.}
\label{fig7}
\end{figure*}

\subsection{Performance Comparison under Varying $\epsilon_s$}

As shown in Fig. \ref{fig9}, we present a performance comparison of PTPPM, PIM, and CTDP under different values of $\epsilon_s$ on the T-Drive and GeoLife datasets. Unlike PIM, which lacks personalization features, our proposed PTPPM achieves dynamic adjustment between privacy and QoS loss by tuning the parameter $E_m$. PTPPM comprehensively considers key factors influencing location sensitivity as well as road network constraints, allowing it to assign more accurate sensitivity levels and privacy budgets for each location. Therefore, compared to CTDP, it demonstrates superior performance in both privacy protection and QoS.
More specifically, as shown in Fig. \ref{fig9}, when $\epsilon _s$ is small, the privacy and QoS loss are large, and at the same time $D\left( \Phi _t \right)$ also affects the privacy and QoS loss, the larger $E_m$, the larger $D\left( \Phi _t \right)$, the larger the privacy and QoS loss.  

Therefore, when the privacy budget $\epsilon_s$ is relatively small (e.g., $\epsilon_s \in (0.1, 3.1)$), a larger value of $E_m$ (e.g., $E_m = 1.5$) is preferred. As shown in Figs. \ref{fig9a} and \ref{fig9c}, under these settings, PTPPM achieves a similar level of privacy as PIM but with significantly lower QoS loss. As illustrated in Figs. \ref{fig9b} and \ref{fig9d}, although PTPPM provides slightly lower privacy compared to PIM, it results in a much smaller QoS loss.
$E_m$ is the lower bound of the attacker's expected inference error; when $E_m$ is large, the privacy and QoS loss of PTPPM are large, but as $\epsilon _s$ increases, the privacy and QoS loss of PIM gradually decrease. So, in the process of increasing $\epsilon _s$, it is necessary to choose an appropriate $\epsilon _s$ to ensure the lower bound of the inference error.

Similarly, compared to CTDP, PTPPM allows users to personalize the parameter $E_m$ based on their individual needs to achieve the desired level of privacy or QoS.
When the user's location is highly sensitive, a larger value of $E_m$ can be selected to enhance privacy protection; conversely, when high QoS is required, a smaller $E_m$ is more appropriate.

\subsection{Impact of Privacy Parameters on Personalized Trajectory Privacy Protection}

In Fig. \ref{fig7}, we explore the impact of varying privacy budgets $\epsilon$ and inference error thresholds $E_m$ on the user's personalized trajectory privacy protection performance. Both parameters exhibit a significant impact on trajectory privacy and QoS loss. As the total privacy budget $\epsilon _s$ for sensitive locations increases, the privacy budget $\epsilon$ allocated to each location increases, as shown in Figs. \ref{fig7a} and \ref{fig7c}. According to $E\left( \Phi _t \right) \ge e^{\epsilon}E_m$, the increase of $\epsilon$ leads to a rapid expansion in the protection area diameter $D\left( \Phi _t \right)$, so when $\epsilon$ is larger than a certain value, the trajectory privacy and QoS loss begin to increase. The inflection points vary based on the $E_m$ settings, with higher $E_m$ values corresponding to greater lower bounds for the attacker's expected inference error, thereby enhancing trajectory privacy protection. However, practical limitations prevent $D\left( \Phi _t \right)$ from perpetually increasing, ultimately capping trajectory privacy and QoS loss at an upper limit.

Figs. \ref{fig7b} and \ref{fig7d} show that the trajectory privacy and QoS loss increase with the increase of $E_m$ with various $\epsilon _s$. For different $\epsilon _s$, the trajectory privacy and QoS loss increase as $E_m$ increases. Given $\epsilon _s$, the privacy budget $\epsilon$ for each location is determined accordingly. When $E_m$ increases, the protected region's $D\left( \Phi _t \right)$ increases, increasing the trajectory privacy and QoS loss. The effect of $\epsilon$ on $D\left( \Phi _t \right)$ is exponential and much higher than that of $E_m$. So when $\epsilon _s$ is set to 4.1, with the increase of $E_m$, $D\left( \Phi _t \right)$ changes significantly, the trajectory privacy and QoS loss increase sharply, and the curve is steep. Since $D\left( \Phi _t \right)$ of the protection region is bounded, the privacy and QoS loss will converge to a finite value. A minimal privacy budget cannot ensure the QoS. It can be observed that regardless of the settings of $E_m$, the user's trajectory privacy and QoS loss reach the upper bound when $\epsilon _s$ is 0.1. By adjusting the privacy parameters, the personalized protection of user trajectory privacy can be realized, which proves the feasibility and effectiveness of PTPPM.
\setlength{\textfloatsep}{10pt plus 1.0pt minus 2.0pt}
 \begin{figure}[!t]
\centering
\subfigure[Size vs. $\delta$]{\includegraphics[width=0.23\textwidth]{{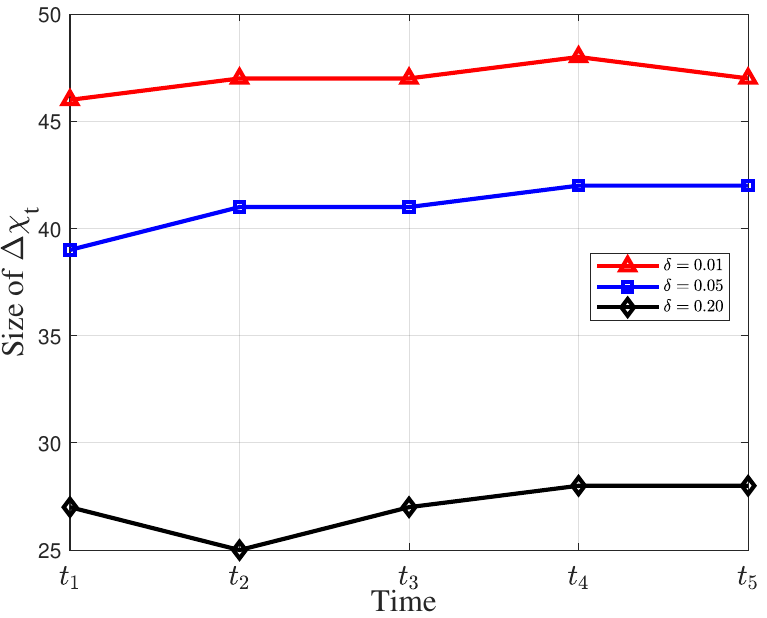}}%
\label{fig6a}}
\hfil
\subfigure[Trajectory Privacy vs. $\delta$]{\includegraphics[width=0.23\textwidth]{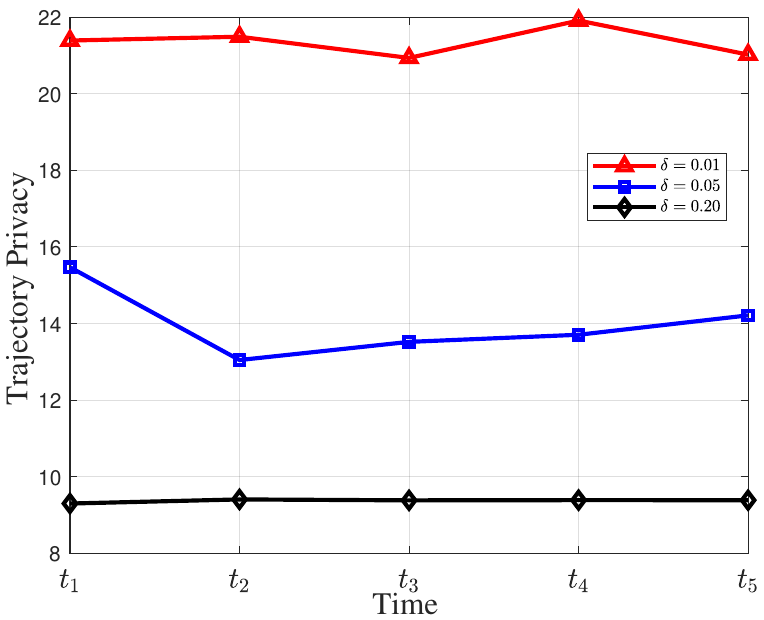}%
\label{fig6b}}
\hfill
\captionsetup[subfigure]{skip=2pt}
\caption{The impact of $\delta$ on $\Delta X_t$ and trajectory privacy.}
\label{fig6}
\end{figure}
Fig. \ref{fig6} illustrates the impact of the parameter $\delta$ on $\Delta X_t$ over consecutive time steps in the trajectory, as well as its effect on privacy. As shown in Fig. \ref{fig6a}, the size of $\Delta X_t$ is primarily influenced by the value of $\delta$; as $\delta$ increases, $\Delta X_t$ becomes smaller because more improbable locations are excluded (truncated). However, an excessively large $\delta$ increases the risk of privacy leakage. As depicted in Fig. \ref{fig6b}, a larger $\delta$ results in lower privacy at each time step. This is because a larger $\delta$ reduces the number of locations within $\Delta X_t$, making it impossible to construct a PLS that satisfies (\ref{PLScondition}), thereby weakening privacy protection. For example, at time step $t_3$, when $\delta = 0.01$, $\Delta X_t$ contains 47 locations and yields a privacy value of 20.94, which is 123\% of the value when $\delta = 0.20$.
\setlength{\textfloatsep}{10pt plus 1.0pt minus 2.0pt}
\begin{figure*}[!t]
\centering
\subfigure[The optimal inference attack at $t_2$]{
\includegraphics[width=0.23\textwidth]{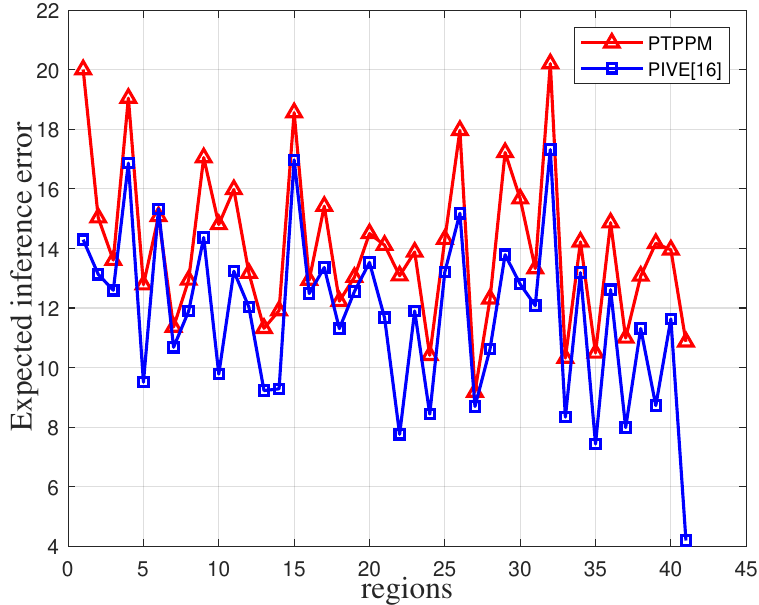} \label{fig8a}
}
\hfill
\subfigure[The optimal inference attack at $t_3$]{
\includegraphics[width=0.23\textwidth]{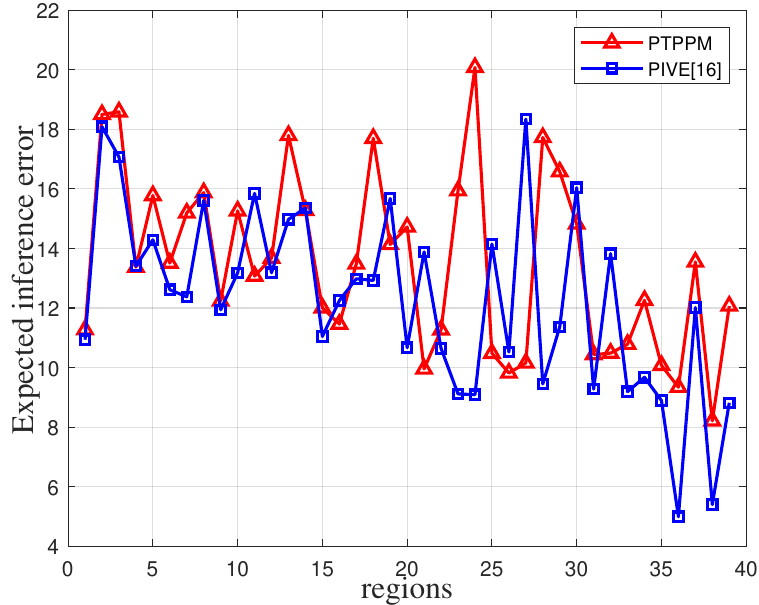} \label{fig8b}
}
\hfill
\subfigure[The Bayesian inference attack at $t_2$]{
\includegraphics[width=0.23\textwidth]{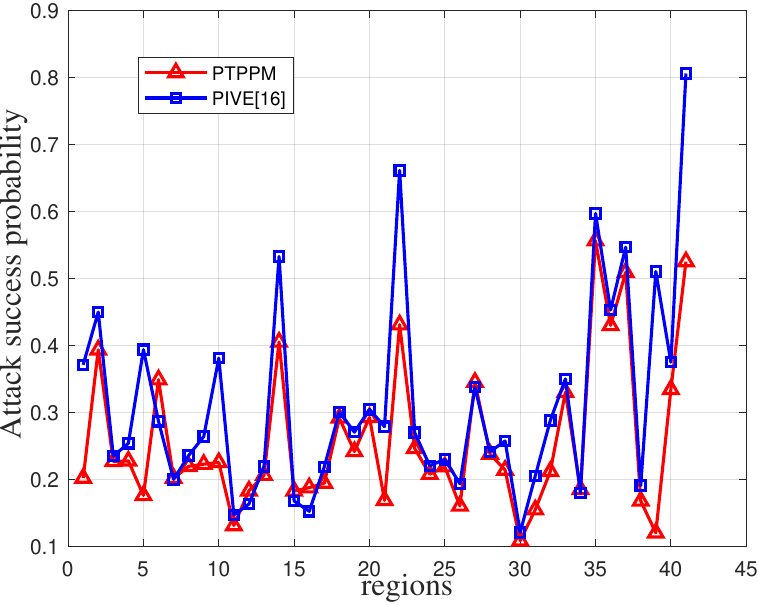}\label{fig8c}
}
\hfill
\subfigure[The Bayesian inference attack at $t_3$]{
\includegraphics[width=0.23\textwidth]{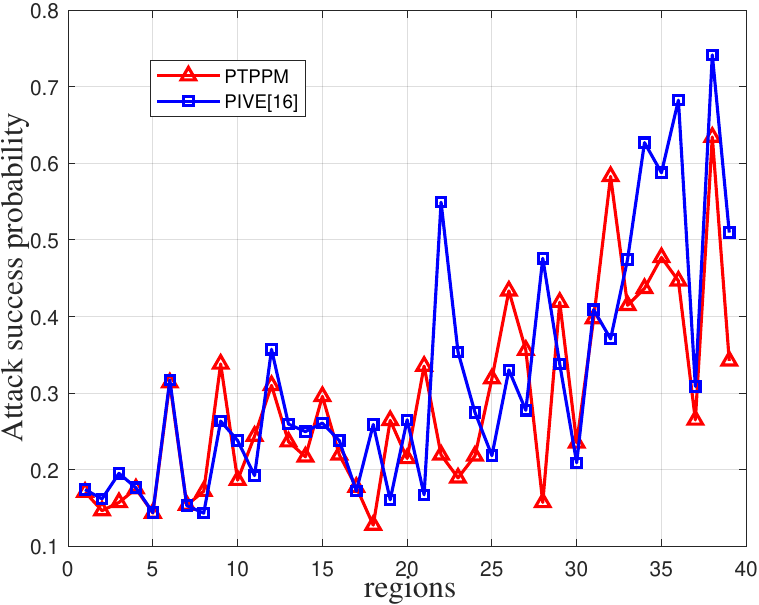}\label{fig8d}
}
\hfill
\captionsetup[subfigure]{skip=2pt}
\caption{Performance against different attack models.}
\label{fig8}
\end{figure*}
\subsection{Performance Against Different Attack Models}
Next, considering spatiotemporal correlation information, we validate the performance advantages of the proposed PTPPM against various types of attackers. We compare PTPPM, which takes the spatiotemporal correlation between different locations on the trajectory into account, with PIVE, which is not considered. First, we find PLS for locations within $\varDelta \chi _{t_2}$ and $\varDelta \chi _{t_3}$ with PIVE and PTPPM, respectively, and afterward, uniformly perturb with the PF mechanism. Comparing the attacker's expected inference error in the optimal inference attack with the attacker's success probability in the Bayesian attack for both schemes.

 As shown in Figs. \ref{fig8a} and \ref{fig8b}, PTPPM makes the attacker's expected inference error larger for 98\% of locations within $\varDelta \chi _{t_2}$ and 72\% of locations within $\varDelta \chi _{t_3}$. This is because PTPPM infers in which locations the user is more likely to appear at the next moment based on the user's current location. That is, the prior probability of these locations increases. According to (\ref{8}), increasing the prior probability will increase $E\left( \Phi _t \right)$. Whereas $E\left( \Phi _t \right) \le D\left( \Phi _t \right)$, PLS will be searched with a larger diameter. So, PLS may contain locations that are farther apart, leading to an increase in the expected inference error of the attacker. As shown in Figs. \ref{fig8c} and \ref{fig8d}, PTPPM makes the attacker's success probability smaller for 83\% of locations within $\varDelta \chi _{t_2}$ and 64\% of locations within $\varDelta \chi _{t_3}$. This is because searching with a larger diameter makes the PLS potentially contain more locations, resulting in a reduced success probability for the attacker.

\subsection{Performance Comparison under Varying QoS Loss}
We conducted a quantitative comparison of the performance of PTPPM, PIVE, PIM, and CTDP in terms of trajectory privacy protection and QoS loss on both the T-Driver and GeoLife datasets, thereby validating the performance advantages and cross-scenario applicability of PTPPM. We take the predefined value in (\ref{Privacy}) as the target QoS loss, where the only variable in the equation is the privacy budget $\epsilon$. By solving this equation, we can obtain the corresponding $\epsilon$ values for PTPPM, PIVE, PIM, and CTDP under the same privacy condition. By substituting (\ref{QoS loss}), the corresponding QoS loss of PTPPM, PIVE, PIM, and CTDP under the same privacy can be calculated.

As shown in Fig. \ref{fig7c}, we analyze how selecting different values of $E_m$ affects the lower bound of QoS loss in PTPPM. Since the lower bounds of QoS loss vary between the T-Drive and GeoLife datasets, Fig. 10(a) and Fig. 10(b) respectively start from QoS = 80 and QoS = 65.
As shown in Fig. \ref{fig10}, PTPPM effectively reduces QoS loss under the same level of privacy protection.
For instance, as illustrated in Fig. \ref{fig10a}, when the privacy protection level is set to 80, PTPPM achieves a QoS loss that is 2.30\% lower than PIVE, 5.56\% lower than CTDP, and 7.61\% lower than PIM. Conversely, in Fig. \ref{fig10b}, when the QoS loss is fixed at 65, PTPPM provides a privacy improvement of 1.10\% over PIVE, 3.50\% over CTDP, and 5.91\% over PIM.
This advantage stems from the fact that, unlike CTDP and PIM, both PTPPM and PIVE incorporate a lower bound constraint on the inference error $E_m$, which enhances system robustness.

Compared to PIVE, PTPPM adopts the PF mechanism to release perturbed locations, reducing the perturbation distance while still satisfying the PLS privacy requirement.
Moreover, since $D$ cannot be increased indefinitely in practical applications, the level of trajectory privacy will eventually reach an upper limit.
These results demonstrate that PTPPM not only provides stronger privacy protection but also better meets users’ QoS requirements.

\setlength{\textfloatsep}{10pt plus 1.0pt minus 2.0pt}
 \begin{figure}[!t]
\centering
\subfigure[Trajectory privacy vs. QoS loss (T-Drive)]{\includegraphics[width=0.23\textwidth]{{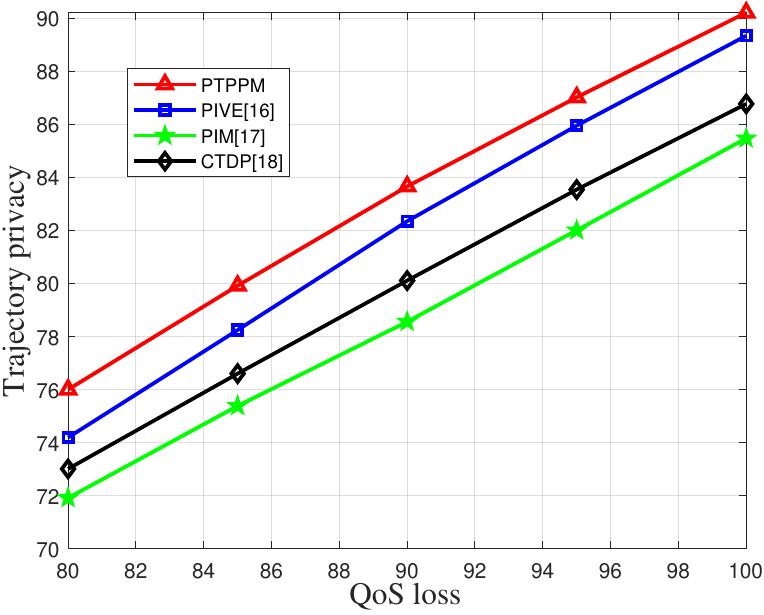}}%
\label{fig10a}}
\hfil
\subfigure[Trajectory privacy vs. QoS loss (GeoLife)]{\includegraphics[width=0.23\textwidth]{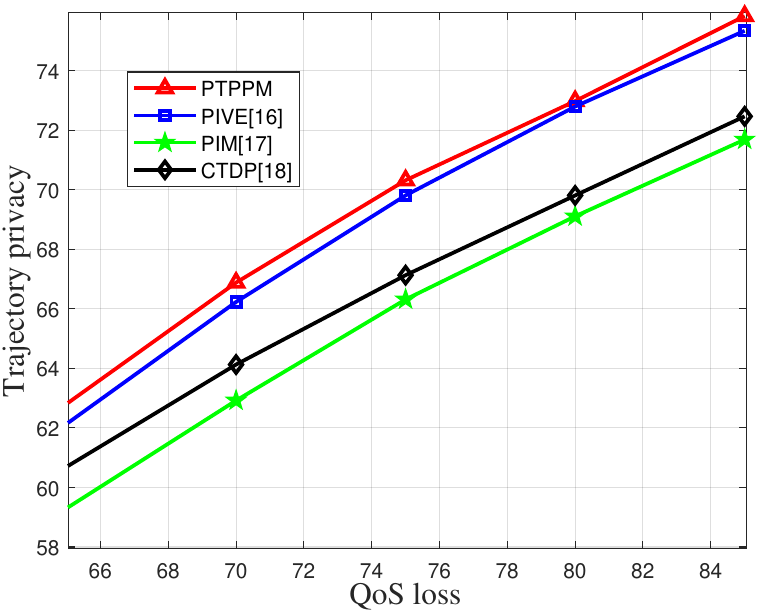}%
\label{fig10b}}
\caption{Performance comparison of different TPPMs under varying QoS loss on T-Drive and Geolife datasets.}
\label{fig10}
\end{figure}
\section{Conclusion}\label{Conclusion}
In this paper, we have proposed a personalized trajectory privacy protection mechanism, PTPPM. This paper has three novel contributions. First, we address the issue of the attacker exploiting the spatiotemporal correlations between different locations to compromise user privacy. To mitigate this threat, we design a robust trajectory privacy protection mechanism. Second, we develop a personalized privacy budget allocation algorithm by integrating a more practical and comprehensive sensitivity assessment approach. We adjust the privacy budget and the expected inference errors bound to achieve personalized privacy protection. By combining the concepts of geo-indistinguishability and distortion privacy, we enhance the system's robustness against location inference attacks while achieving a personalized approach. Finally, we propose a novel perturbation mechanism, Permute-and-Flip, which releases perturbed locations with smaller perturbation distances to better balance the trajectory privacy and QoS. Simulation results demonstrate that PTPPM can effectively achieve personalized trajectory privacy protection, achieving a significant performance advantage over PIVE under different attack models. For instance, in 98\% of locations within $\varDelta \chi _{t_2}$, PTPPM results in a larger attacker's expected inference error compared to PIVE. Moreover, PTPPM reduces the attacker's success probability in 83\% of the locations within $\varDelta \chi _{t_2}$.

\balance
\bibliography{ref}
\bibliographystyle{ieeetr}

\end{document}